\documentclass[a4,12pt]{article}
\UseRawInputEncoding
\usepackage{bm}
\usepackage[all]{xy}
\usepackage{graphicx}
\usepackage{verbatim}
\usepackage{wrapfig}
\usepackage{makeidx}
\usepackage{amscd}
\usepackage{url}
\usepackage{comment}
\usepackage{enumerate}
\usepackage{here}
\usepackage{latexsym}
\usepackage{array}
\usepackage{booktabs}
\usepackage{multirow}
\usepackage{mathrsfs}
\usepackage{geometry}
\usepackage{fancyhdr}
\renewcommand{\title}[1]{

\begin{center} \Large \bf #1 \end{center}
}

\renewcommand{\author}[2]{
 \begin{center} #1  \vspace{3mm} \\
  #2 \\
 \end{center}
\addvspace{\baselineskip}
}

\usepackage{amssymb}
\usepackage{amsmath}

\usepackage{amsthm}
\newtheorem{theorem}{Theorem}[section]
\newtheorem{proposition}[theorem]{Proposition}
\newtheorem{corollary}[theorem]{Corollary}

\newtheorem{example}[theorem]{Example}
\theoremstyle{definition}
\newtheorem{definition}[theorem]{Definition}

\theoremstyle{remark}
\newtheorem*{rem}{Remark}

\makeatletter
\@addtoreset{equation}{section}

\makeatother

\begin{document}
\baselineskip 5mm
\title{Category of Quantizations and Inverse Problem}
\author{Akifumi Sako}
{
Tokyo University of Science,\\ 1-3 Kagurazaka, Shinjuku-ku, Tokyo, 162-8601, Japan
}
\noindent
\vspace{1cm}

\abstract{ 
We introduce a category composed of all
quantizations of all Poisson algebras. By the category, we can
 treat in a unified way the various quantizations for all Poisson algebras
and develop a new classical limit formulation.
This formulation proposes a new 
method for the inverse problem, that is, 
the problem of finding the classical limit from a quantized space. 
Equivalence of quantizations is defined by using this category, 
and the conditions under which the two quantizations are equivalent are investigated.
Two types of the classical limits are defined
as the limits in the context of category theory，
and they are determined by giving a sequence of objects.
Using these classical limits, 
we discussed the inverse problem of determining 
the classical limit from some 
noncommutative Lie algebra.
From a Lie algebra, we construct a sequence of quantized spaces, 
from which we determine a Poisson algebra. 
We also present a method to obtain this sequence of quantizations 
from the principle of least action by using matrix regularization. 
Apart from the above category of quantizations of
all Poisson algebras, 
we also introduce a category of quantizations 
of a fixed single Poisson algebra. 
In this category, the other classical limit is defined, and 
it is automatically determined for the category.
}
%
%
%
\section{Introduction}

In M-theory and string theory, 
it has been proposed for a long time
that matrix models give their constructive formulation \cite{BFSS,IKKT3}.
Noncommutative manifolds are obtained as classical solutions 
of such matrix models. (See for example \cite{IKKT2,fuzzya} and references therein.)
Assuming that the universe as a noncommutative manifold is realized, 
the space we inhabit is perceived as a smooth manifold, 
at least to an approximation. 
Therefore, it is important to obtain smooth manifolds 
as classical limits of noncommutative manifolds.

In many cases, noncommutative geometry is 
some kind of a noncommutative 
deformation of a smooth manifold. 
In other words, it is often made as a quantization of a smooth manifold.
So, it is often called an inverse problem
to construct a usual commutative manifold or commutative algebra 
as a classical limit of a noncommutative manifold.

Taking the classical limit of a noncommutative manifold 
is often fraught with difficulty.
The images of the quantization maps 
have information of the classical geometry of manifolds, 
but the algebras generated from them lose
information of the original geometry.
For example, as a well-known phenomenon, 
matrix regularization generates the same matrix algebra 
whether a two-dimensional torus is transformed into a fuzzy torus 
or a two-dimensional sphere into a fuzzy sphere.
(See Appendix \ref{appendix_fuzzy} and
\cite{deWitHoppeNicolai,arnlind,matrix1,fuzzy1, MadoreText}.)
Put another way, when reading a classical geometry (Poisson algebra)
from a matrix algebra, 
the Poisson algebra obtained depends on what classical limit is taken
\cite{berezin1,bordemannA,Chu:2001xi}.
Various approaches have been taken to the inverse problem 
of how to extract geometric properties from matrix algebras
\cite{shimada,berenstein,Schneiderbauer,ishiki1,ishiki2,asakawa}.
Therefore, a framework including any Poisson algebra and its quantizations is important in considering such issues to be investigated in a unified manner.\\

Another important lesson implied by the difficulty described above 
is that the commutative limit cannot be identified 
by information in the algebra obtained as the target of a quantization map.
Rather, information about the modules or vector spaces that are images of
quantization maps before their generating algebras is important to identify the classical limits. This tendency can be read from, for example, fuzzy spheres and fuzzy torus in Appendix \ref{appendix_fuzzy}.
Note that an image of a quantization map is not an algebra but a module 
in general since a quantization map is not an algebra homomorphism but a 
module homomorphism (linear map).
Furthermore, the classical limit depends on what classical limit is taken
as already mentioned.
Accordingly, it is suggested that 
we should prepare a sequence of modules 
to determine a Poisson algebra as a classical limit, not just one of them.
\\

In light of the above, the first purpose of this study is 
to construct a category composed of the all 
quantizations of all Poisson algebras.
This is actually performed in Section \ref{sect2+1}.
We denote this category by $QW$ (Quantum World),
and its strict definition is given in Section \ref{sect2+1}.
Quantization here refers to a noncommutative deformation 
in the sense of approximate canonical quantization of a Poisson algebra.
The definition and properties of this 
quantization are given in Section \ref{sect2}.
Quantizations in this paper are not necessarily 
identified with them in the sense of quantum mechanics, 
nor in the sense of spectral discretization, 
but in the sense of many noncommutative geometries.
The morphisms that define this $QW$ consist of Poisson morphisms, 
quantization maps, and algebra homomorphisms whose domain is restricted. 
Their compositions are constructed 
in such a way that they satisfy the axioms 
of the category without contradiction.
Furthermore, in this context, it is also important 
to show that quantization can be properly classified. 
To this end, we construct a functor from $QW$ to some 
comma category and show that 
the isomorphism of the object of the comma category 
gives the isomorphism pair of the Poisson algebra 
and the algebra generated by its quantized space.
$QW$ allows any quantization of any Poisson algebra 
to be treated as a single framework.
\\

The second purpose is to define the classical limit 
within the framework of this category of the all quantizations $QW$
in a non-traditional way.
Of course, the question of what classical limit is best 
within the $QW$ framework depends on what classical limit is needed. 
In this paper, we propose three types of classical limits and investigate their properties by considering concrete examples 
in Section \ref{sect_ClassicalLimit}.
We define the classical limit 
using the limit in the category theory or a slight modified limit.
The first two types are defined by using sequence of objects (modules) in $QW$
and morphisms between them.
One is a naive adaptation of the category theory limit to $QW$. 
When that category theory limit corresponds to a Poisson algebra, 
it is defined to be the classical limit. 
To define another classical limit, it is necessary to modify the definition 
of the limit of category theory.
We impose universality within the framework that the vertex of the cone is a Poisson algebra.
The classical limit so defined is called the weak classical limit.
These limits depends on how to choose the sequence of objects, and the 
existence of the limits are not guaranteed.
As concrete examples, we investigate the classical limits 
for a sequence of deformation quantizations of a polynomial ring 
and for a sequence of fuzzy spheres.
By considering the all quantizations as a single category, 
we can make a structure that
a sequence of modules as a sequence of objects in the category 
determines a Poisson algebra as the classical limit from all 
Poisson algebras.
To make the sequence of objects, we use matrix regularization as an example.
Furthermore, developing this argument 
we study how to obtain the classical limit by using 
principle of least action in Section \ref{sect5}.

The third type of the classical limit is 
made by focusing a single Poisson algebra.
When we choose a Poisson algebra, this classical limit is determined
automatically.
It is constructed in the framework of a category 
made from $QW$ with a fixed Poisson algebra whose
objects are quantization maps from the Poisson algebra.
We prove that the classical limit coincides with the fixed Poisson algebra.
We call this classical limit the strong classical limit.
\\

We would like to make a few comments on recent 
previous studies that discuss quantization 
as noncommutative deformation of 
a Poisson algebra by using category theory approach.
The construction of a quantization category 
was done by \cite{Gohara:2019kkd,Jumpei:2020ngc}.
The quantization category in \cite{Gohara:2019kkd,Jumpei:2020ngc}
was pioneering in the sense that it was a category-theoretic version of quantization, but it contained two problems.
The first problem is that the category in \cite{Gohara:2019kkd,Jumpei:2020ngc} 
is constructed as a whole of quantization 
with respect to a single fixed Poisson algebra. 
As already mentioned, this is a weakness for the inverse problem
because it is important to treat the all quantizations
for the all Poisson algebras, in a unified manner.
The second problem is that the quantization maps treated in 
\cite{Gohara:2019kkd,Jumpei:2020ngc}
are maps from some Poisson algebra to Lie algebras.
Target spaces of the maps do not necessarily reflect the reality of the image of the mapping. A Lie algebra that is a target space of some quantization map
may have lost some of their original information 
of the image of quantization map, as already mentioned.
So, to classify quantization maps, we need more complex structure in the formulation.
On the other hand, the category constructed in this article conquers these two problems.
\\

This paper is organized as follows.
In Section \ref{sect2}, we study quantization maps.
Using the properties of quantization maps, we
construct a category of the all 
quantizations of all Poisson algebras $QW$ in Section \ref{sect2+1}.
In Section \ref{sect3}, a way to classify the quantization maps
is established by using the context of $QW$.
In Section \ref{sect_ClassicalLimit},
we discuss how to find the classical limit in $QW$.
Three types of classical limits are discussed.
Examples of these classical limits are studied.
In Section \ref{sect5},
the invers problem is discussed in the context of the classical limit
defined in Section \ref{sect_ClassicalLimit}.
Especially, a way to obtain the classical limit
by the principle of least action with a sequence of actions
is proposed.
We make a summary of this article in Section \ref{sect6}.
A symbol like Landau notation, 
used throughout this paper, is defined in Appendix \ref{ap1}.
Many new definitions appear in this paper, 
and a list of the symbols is available in Appendix \ref{ap2}.
We often use a fuzzy sphere or a fuzzy torus to make examples.
So, we make a short introduction to them in Appendix \ref{appendix_fuzzy}.


\section{Quantization Maps}\label{sect2}
We define the quantization maps in this section.
To formulate them, we use $R$-modules and $R$-algebras.
$R$ denote a fixed commutative ring over ${\mathbb C}$, in this paper.
Only finitely generated $R$-modules and $R$-algebras are considered.
As the simplest case, we can employ ${\mathbb C}$ as $R$
in all the discussions in this paper.

\subsection{Definition of Weak Quantization Maps}
We define a category of Poisson algebras whose morphisms are restricted 
into surjective Poisson morphisms as follows.
\begin{definition}
Let $(A, \cdot_A ,\{ ~ , ~ \}_A ) $ and $( B, \cdot_B , \{ ~ , ~ \}_B)$ be
Poisson algebras over $R$.  When a linear $\phi_{A, B} : A \rightarrow B $  satisfies
\begin{align}
\phi_{A, B}( f \cdot_A g ) &= \phi_{A, B}(f) \cdot_B \phi_{A, B}(g) \\
\phi_{A, B}(\{ f , g \}_A )&= \{ \phi_{A, B}(f) ,  \phi_{A, B}(g) \}_B ,
\end{align}
for any $ f,g  \in A$, $\phi_{A, B}$ is called a Poisson morphism.
${\mathscr Poisss}$ is a category whose objects are Poisson algebras 
over $R$ and
its morphisms are ``surjective" Poisson morphisms.
\end{definition}
The reason why the morphisms were restricted to surjective maps 
will become clear later.
Often the subscripts of Poisson brackets and the symbol for the product 
$\cdot_A$ and so on are omitted in the following.
When we put $R ={\mathbb C}$, Poisson algebras are 
usual ones.
All examples appearing in this paper are the cases 
that $R$ is fixed to  ${\mathbb C}$.

\begin{definition}[Weak Quantization map $wkQ$]\label{WQ}
Let ${A}$ be a Poisson algebra over $R$, and
let $T_i$ be an $R$-module that is given by a subset of some 
$R$-algebra $(M , *_M)$.
If an $R$-module homomorphism (linear map)  $t_{Ai} \in Hom_{R\mbox{\tiny -Mod}} ({A} , T_i) $ 
equips a constant $\hbar(t_{Ai}) \in {\mathbb C}$ and 
satisfies 
\begin{align}\label{lie}
[t_{Ai}(f),t_{Ai}(g)]_M =\sqrt{-1}\hbar (t_{Ai}) t_{Ai}(\{f,g\})
+\tilde{O}(\hbar^{1+\epsilon} (t_{Ai})) \quad 
(\epsilon >0 )
\end{align}
for arbitrary $f, g \in {A}$, where $[a ,b]_M := a *_M b - b *_M a $,
we call $t_i$ a weak quantization map.
For the case that 
$[t_{Ai}(f),t_{Ai}(g)]_M= 0 $ and $t_{Ai}( \{ f,g \} )=0 $
for $\forall f,g \in A $ , we put $\hbar(t_{Ai})=0$. 
$\tilde{O}$ is defined in the Appendix \ref{ap1}.
We denote the set of all weak quantization maps by $wkQ$:
\begin{align*}
wkQ := \bigsqcup_{\substack{  A \in {\mathscr Poisss} \\ 
T_i \subset M \in R\mbox{\tiny -alg}}} 
\{ t_{Ai} \in Hom_{R\mbox{\tiny -Mod}}(A, T_i) ~ | ~ t_{Ai} ~\mbox{satisfies } ~
(\ref{lie}) \}
\end{align*}
\end{definition}

Note that the target of $t_{Ai}$ is not an algebra in general,
but it is a subset of an algebra $M$, so the left-hand side of (\ref{lie}) is defined through the product in $M$.
We call $\hbar(t_{Ai}) \in {\mathbb C}$ a noncommutativity parameter.
As we will see in Proposition \ref{prop7},
the $\hbar $ can take any value of ${\mathbb C}^\times$
for fixed $A$ and $T_i$ if $\hbar(t_{Ai}) \neq 0$.
In the following, $*_M$ is omitted as appropriate.
For simplicity, $[a, b]_M$ is abbreviated to $[a,b]$
when there is no risk of confusion.
For the case  $R={\mathbb C}$, $t_{Ai} \in Hom_{{\mathbb C}}(A, T_i) $
is a linear map between ${\mathbb C}$-vector spaces with the condition (\ref{lie}).
\\

Definition \ref{WQ} includes many kinds of quantizations.
For example, matrix regularization \cite{matrix1,arnlind},
fuzzy spaces \cite{fuzzy1}, and
Berezin-Toeplitz quantization 
\cite{berezin1,bordemannA,berezin2}
which have original ideas of the matrix regularization,
satisfy Definition \ref{WQ}.
In fact, the fuzzy sphere appears frequently as an example
in this paper.
In addition, the strict deformation quantization introduced by Rieffel 
\cite{rieffel1,rieffel2,strict2},
the prequantization \cite{prequantization0,prequantization1,prequantization2},
and Poisson enveloping algebras \cite{oh,oh2,um,can}
are also in $wkQ$.
These facts are derived immediately from these definitions.
(In \cite{Gohara:2019kkd,Jumpei:2020ngc},
we can see organized discussions about these quantization maps. 
The conditions for $wkQ$ is a part of the definition of 
pre-$\mathscr{Q}$ in \cite{Gohara:2019kkd,Jumpei:2020ngc}.)

\subsection{Properties for Weak Quantization Maps}
For the later part, we summarize basic properties of 
$R$-algebra homomorphisms used in this paper, and 
we derive some propositions of the weak quantization maps.
In order to make this article a self-contained article, 
proofs of basic matters are not omitted.

\begin{proposition}\label{prop1}
Let $(M , *_M)$ and $(N, *_N)$ be $R$-algebras, and let $T$ be a generating set of $M$, i.e. 
$\langle T \rangle = M$.
For any $R$-algebra homomorphisms $h_{MN}, h'_{MN}$ from $M$ to $N$,
\begin{align}
h_{MN}|_T = h'_{MN}|_T ~ \Leftrightarrow h_{MN}= h'_{MN},
\end{align}
where $h_{NM} |_T$ and $h'_{NM}|_T$ are $h_{MN}$ and $ h'_{MN}$ whose domains are restricted
to $T$, respectively.
\end{proposition}

\begin{proof}
For a multi index $I_k = (i_1 , \cdots , i_k)$, we denote 
$a_{i_1}*_M \cdots *_M a_{i_k}$ by $a_{I_k}$, where
$a_{i_j} \in T$.
Since $T$ generates $M$, any $b \in M$ can be written as
a polynomial
\begin{align*}
b= \sum_k \sum_{I_k} b_{I_k} a_{I_k} ~ \qquad ( b_{I_k} \in R ).
\end{align*}
If $h_{MN}|_T = h'_{MN}|_T$, $h_{MN}(a_{i_j}) = h'_{MN}(a_{i_j})$
for each $a_{i_j}$, and
\begin{align*}
h_{NM}(b) &= \sum_k \sum_{I_k} b_{I_k} h_{NM}(a_{I_k}) \\
&= \sum_k \sum_{I_k} b_{I_k} h'_{NM}(a_{I_k})= h'_{NM}(b).
\end{align*}
The proof for the opposite direction is trivial.
\end{proof}


\begin{proposition}\label{prop2}
Let $M$ and $\langle T \rangle $ be $R$-algebras,
where $T $ is a subset of $\langle T \rangle$ that generates algebra 
$\langle T \rangle$.
For any $R$-algebra homomorphism $h : \langle T \rangle \rightarrow M$
\begin{align}
\langle h(T) \rangle = h (\langle T \rangle ).
\end{align}
\end{proposition}

\begin{proof}
For any $m \in \langle h(T) \rangle$ can be written as
a polynomial
\begin{align*}
m= \sum_k \sum_{I_k} m_{I_k} e_{I_k} ~ \qquad ( m_{I_k} \in R ),
\end{align*}
where 
$e_{I_k}=e_{i_1}* \cdots * e_{i_k}$ with $e_{i_j}\in h (T)$.
There exists $e'_{i_j} \in T$ such that
$e_{i_j}= h (e'_{i_j})$ for each $e_{i_j}$. Then,
\begin{align*}
m &= \sum_k \sum_{I_k} m_{I_k} h (e'_{i_1})* \cdots * h (e'_{i_k}) 
\\
&= h( \sum_k \sum_{I_k} m_{I_k} e'_{i_1} * \cdots * e'_{i_k} ).
\end{align*}
\end{proof}


\begin{proposition}\label{prop3}
Let $A$ be a Poisson algebra.
Let $t_i : A \rightarrow M_i$ be a weak quantization map whose 
image $t_i(A)$ generates
an $R$-algebra $M_i$, i.e. $\langle t_i(A) \rangle = M_i$.
For an arbitrary $R$-algebra $M_j$ and any $R$-algebra homomorphism 
$h_{ij} : M_i \rightarrow M_j$, 
an $R$-module homomorphism $t_j : A \rightarrow t_j(A) \subset M_j$ 
defined by
\begin{align}\label{t_j}
t_j &= h_{ij}\circ t_i , 
\end{align}
satisfies 
\begin{align}\label{prop3_1}
[ t_j (f) , t_j(g)]_{M_j} &=\sqrt{-1}\hbar (t_i) t_j(\{f , g \})+\tilde{O}(\hbar^{1+\epsilon} (t_i))  ~~\ ( \forall f, g \in A ) .
\end{align}
\end{proposition}

\begin{proof}
For any $f, g \in A$, from (\ref{t_j}) and 
the fact that $h_{ij}$ is an
$R$-algebra homomorphism, 
\begin{align*}
&[ t_j (f) , t_j(g)]_{M_j} \\
&=
(h_{ij}|_{t_i(A)} \circ t_i )(f) *_{M_j} (h_{ij}|_{t_i(A)} \circ t_i )(g)
-(h_{ij}|_{t_i(A)} \circ t_i )(g) *_{M_j} (h_{ji}|_{t_i(A)} \circ t_i )(f)
\\
&= h_{ij} (  t_i (f) *_{M_i} t_i (g)-  t_i (f) *_{M_i} t_i (g))\\
&= h_{ij} (  [t_i (f),~ t_i (g)]_{M_i})
\end{align*}
Since $t_i$ is in $wkQ$, the above equation is written as
\begin{align*}
[ t_j (f) , t_j(g)]_{M_J} 
&= h_{ij} ( \sqrt{-1}\hbar (t_i) t_i(\{f , g \})+\tilde{O}(\hbar^{1+\epsilon} (t_i)))\\ 
&=  \sqrt{-1}\hbar (t_i)  t_j (\{f , g \}) + \tilde{O}(\hbar^{1+\epsilon} (t_i))).
\end{align*}
Here we use Proposition \ref{propA2} in Appendix \ref{ap1}.
\end{proof}

Note that (\ref{prop3_1}) does not mean that
$t_j$ is a weak quantization map with 
$\hbar (t_j)= \hbar (t_i)$.
For example,  $\hbar (t_i)\neq 0$ case, 
if $t_j( A )$ generates commutative algebra, (\ref{prop3_1}) 
means $\{f , g \}$ is in ${\rm Ker }~ t_j$.
In this case $\hbar (t_j)= 0$
by the definition of the weak quantization.
\\

\begin{corollary}\label{coro3}
Let $A$ be a nontrivial Poisson algebra, i.e.
$\exists f,g \in A $ such that $\{f , g \} \neq 0$.
Let $t_i, t_j, h_{ij}, M_i, M_j$ be them in 
Proposition \ref{prop3}.
If $\langle t_i (A) \rangle$ is noncommutative and if
$\langle t_j (A) \rangle =  \langle h_{ij}\circ t_i  (A)\rangle \subset M_j$
is commutative, then 
\begin{align}
\{ A , A\} := \{~ \{f, g \} \in A ~| ~ f,g \in A \} \subset {\rm Ker} ~t_j,
\end{align}
is not $\{ 0 \}$. 
$t_i ( \{ A , A\} ) $ is not $\{ 0 \}$, too, and
\begin{align}
t_i ( \{ A , A\} ) \subset {\rm Ker} ~h_{ij}.
\end{align}
\end{corollary}

\begin{proof}
When $\langle t_j (A) \rangle$ is a commutative algebra,
for any $f,g \in A$,
\begin{align}
0 &= [ t_j (f) , t_j (g)]=
i \hbar (t_i) t_j (\{ f,g \} ) + \tilde{O}(\hbar^{1+\epsilon}) .
\end{align}
Since 
$\langle t_i (A) \rangle$ is noncommutative, 
existence of $f,g \in A$ such that $[ t_i (f) , t_i(g) ]\neq 0$ 
derives $\hbar (t_i)\neq 0$.
Then we get 
$t_j (\{ f,g \} )= h_{ij}\circ t_i (\{ f,g \} )
=0$ for $\forall f, g \in A$.
$\hbar (t_i)\neq 0$ also means that
there exist $f, g \in A$ such that
$t_i (\{ f,g \} )\neq 0$, and
$t_i ( \{ A , A\} ) \subset {\rm Ker} ~h_{ij}$ is also obtained.
\end{proof}

\begin{proposition}\label{thm3}
Suppose that a noncommutativity parameter of a
weak quantization map 
$t_i : A \rightarrow t_i(A) \subset M_i$ is non-zero, 
i.e. $\hbar(t_i ) \neq 0$, where $M_i $ is some noncommutative $R$-algebra.
For an $R$-algebra homomorphism $h_{ij} : M_i \rightarrow M_j$ and
$ t_j := h_{ij}\circ t_i $ , $\langle t_j (A) \rangle$
is a commutative algebra if and only if 
\begin{align*}
\{ [ t_i(f) , t_i(g) ] ~| ~   f,g \in A \} \subset {\rm Ker}~ h_{ij} .
\end{align*}
\end{proposition}

\begin{proof}
For any $ f,g \in A $ with $[ t_i(f), t_i(g) ] \neq 0$,
\begin{align*}
h_{ij}([ t_i(f), t_i(g) ])= [ h_{ij}\circ t_i(f), h_{ij}\circ t_i(g) ]
= [ t_j(f), t_j(g) ].
\end{align*}
Therefore, if $\langle t_j (A) \rangle$ is commutative,
$\forall [ t_i(f) , t_i(g) ] \in {\rm Ker}~ h_{ij}$ .\\

Conversely, suppose that $[ t_i(f) , t_i(g) ]$ is in ${\rm Ker}~ h_{ij}$
for $\forall f,g \in A$.
$\forall m_l \in \langle t_j (A) \rangle$ is expressed as
$\displaystyle m_l
= \sum_k \sum_{I_k} m_{lI_k} a_{I_k}$,
where $a_{I_k}= a_{i_1}\cdots a_{i_k}, 
~a_{i_p}\in {\rm Im}~ h_{ij}$ and $m_{lI_k} \in R$.
There exists $e_{i_p} \in t_i(A)$
such that $a_{i_p} = h_{ij}(e_{i_p})$.
Any commutator $[ m_l , m_p ]$ is written as
\begin{align*}
[ m_l , m_p ]=
\sum_{k,q} \sum_{I_k, J_q} m_{l I_k} m_{p J_q}
[a_{I_k} , a_{J_q} ], 
\end{align*}
and 
$\displaystyle [a_{I_k} , a_{J_q} ]
= \sum_{1\le s \le k} \sum_{1\le n \le q} 
(\cdots [ a_{i_s} , a_{j_n} ] \cdots )$. So,
\begin{align*}
[ m_l , m_p ]&=
\sum_{k,q} \sum_{I_k, J_q} m_{l I_k} m_{p J_q}
\sum_{1\le s \le k} \sum_{1\le n \le q} 
(\cdots [ a_{i_s} , a_{j_n} ] \cdots ).
\\
&= \sum_{k,q} \sum_{I_k, J_q} m_{l I_k} m_{p J_q}
\sum_{1\le s \le k} \sum_{1\le n \le q} 
(\cdots h_{ij}([ e_{i_s} , e_{j_n} ]) \cdots ).
\end{align*}
Since $[ e_{i_s} , e_{j_n} ] \in {\rm Ker}~h_{ij}$,
we find $[ m_l , m_p ]=0$.
\end{proof}

To avoid commutative algebras from quantum 
theories, we will restrict $R$-algebra
homomorphisms to injective ones later.

We found that
$ t_j = h_{ij}\circ t_i $ in Proposition \ref{prop3} 
is a weak quantization map. 
$\langle t_j (A) \rangle =  \langle h_{ij}\circ t_i \rangle \subset M_j$,
${\rm Im} h_{ij} $, and $M_j$ are not always equal. 
So, ``when $\langle  h_{ij}\circ t_i (A) \rangle $ generates  $M_j$?''
is a natural question.
The following discussion derives a useful tool to judge this problem.


\begin{proposition}\label{prop4}
Let $A,B,C \in ob ( \mbox{$R$-Mod} ) $ be  $R$-modules satisfying 
$\langle  A \rangle = M_i$,
$\langle  B \rangle = \langle  C \rangle = M_j$, where $M_i$ and $M_j$ 
are $R$-algebras.
If an $R$-algebra homomorphism $h_{ji} : M_j \rightarrow M_i$
satisfies $h_{ji} |_B (B)= A$, where  $h_{ji} |_B$ is defined
by restriction of its domain to $B$, then 
the image of $C$ by $h_{ji}$ generates $M_i$, i.e.
$\langle  h_{ji} |_C (C) \rangle = M_i$.
\end{proposition}

\begin{proof}
From $\langle  B \rangle = \langle  C \rangle = M_j$,
an arbitrary $b \in B$ is expressed by elements of $C$.
By $h_{ji} |_B (B)= A$ and $\langle  A \rangle = M_i$,
$h_{ji} |_B (B)$ generates $M_i$.
Then $\langle  h_{ji} |_C (C) \rangle = M_i$ is given.
\end{proof}


We introduce a sequence of weak quantization maps 
$t_{\mu j}^{(l)}: A^\mu \rightarrow t_{\mu j}^{(l)}(A^\mu ) \subset M_j,
 (l\in I_{j}, \mu=1,2, \cdots ) $
that generate an $R$-algebra $M_j$
\begin{align}
M_j &= \langle t_{\mu j}^{(l)}({A}^\mu ) \rangle = \cdots 
=  \langle t_{\mu j}^{(n)}({A}^\mu ) \rangle \notag =\cdots \\
&= \langle t_{\nu j}^{(l)}({A}^\nu ) \rangle = \cdots 
=  \langle t_{\nu j}^{(n)}({A}^\nu ) \rangle = \cdots \notag \\
& \vdots \mbox{ \hspace{3cm} },
\end{align}
where ${A}^\mu , {A}^\nu , \cdots \in ob( {{\mathscr Poisss}})$ and
$I_j$ is an index set.

\begin{proposition}\label{prop5}
Let ${A}^\mu , {A}^\tau $ be Poisson algebras.
Let $t_{\mu j}^{(k)}:{A}^\mu \rightarrow t_{\mu j}^{(k)}(A^\mu )$ and 
$t_{\mu i}^{(l)}: {A}^\mu \rightarrow t_{\mu i}^{(l)}(A^\mu )$
be weak quantization maps 
such that 
$\langle t_{\mu j}^{(k)}({A}^\mu ) \rangle =M_j$ , 
$\langle t_{\mu i}^{(l)} ({A}^\mu ) \rangle =M_i$.
If an $R$-algebra homomorphism $m_{ji} : M_j \rightarrow M_i$
satisfies commutativity $m_{ji} \circ t_{\mu j}^{(k)} = t_{\mu i}^{(l)}$,
then $\forall t_{\tau j}^{(n)} \in wkQ$ 
such that 
$\langle t_{\tau j}^{(n)}({A}^\tau ) \rangle =M_j$,
$t_i : = m_{ji} \circ t_{\tau j}^{(n)} : A^{\tau} \rightarrow t_i(A^{\tau})
\subset M_i$
satisfies
\begin{align}
[ t_i (f) , t_i(g)]_{M_i} &=\sqrt{-1}\hbar (t_{\tau j}^{(n)}) t_i(\{f , g \})+\tilde{O}(\hbar^{1+\epsilon} (t_{\tau j}^{(n)}))  ~~\ ( f, g \in A^\tau ) ,
\end{align}
and
\begin{align}
\langle t_{i} ({A}^\tau ) \rangle =M_i .
\end{align}
\end{proposition}
\begin{align*}
\xymatrix{
A^\mu \ar[r]^{t_{\mu j}^{(k)}}\ar@/_48pt/[dddr]_{t_{\mu i}^{(l)}\mbox{\ {}}} & 
t_{\mu j}^{(k)} (A^\mu) \ar@{^{(}-_>}[r]
& M_j \ar[dd]^{m_{ji}}\\
A^\tau \ar[r]^{t_{\tau j}^{(n)}}\ar[rd]^{t_i} & {t_{\tau j}^{(n)}} (A^\tau) 
\ar@{^{(}-_>}[ru]
& {}\\
{}& t_i(A^\tau) \ar@{^{(}-_>}[r] & M_i \\
{}& t_{ \mu i}^{(l)}(A^\mu ) \ar@{^{(}-_>}[ru]& {} 
}
\end{align*}

\begin{proof}
This proposition follows from Proposition \ref{prop3} and 
Proposition \ref{prop4}.

\end{proof}

\begin{rem}
Note that the existence of Poisson maps between 
${A}^\mu $ and ${A}^\tau $ are not requested in Proposition
\ref{prop5}.
\end{rem}
\begin{example}\label{ex_T2_S2}
As an example, let us consider fuzzy spheres and fuzzy tori.
These are obtained by matrix regularization for $S^2$ and 
$T^2$, and
their processes are summarized
in Appendix \ref{appendix_fuzzy}.
$\mathcal{A}$ and $\mathcal{B}$ denote the Poisson algebras defined
as a set of functions on $S^2$ and 
$T^2$, respectively.
$t_k:\mathcal{A}\to T_k := t_k (\mathcal{A})\subset M_k={\rm Mat}_{k}({\mathbb C})$ and
$q_k:\mathcal{B}\to Y_k := q_k (\mathcal{B})\subset M_k={\rm Mat}_{k}({\mathbb C})$ are
defined as weak quantization maps.
The details are given in  Appendix \ref{appendix_fuzzy}.
It is known that both $t_k (\mathcal{A})$ and $q_k (\mathcal{B})$ 
generate the same matrix algebra ${\rm Mat}_{k}$.

Let introduce a map 
$t_{k \oplus k} : \mathcal{A}\to T_k \oplus T_k
\subset M_k \oplus M_k$
as
\begin{align}
t_{k \oplus k} (f) = 
\left(
\begin{array}{c|c}
t_k(f) &0 \\
\hline
0 & t_k(f)
\end{array}
\right) \qquad (f \in \mathcal{A}) .
\end{align}
This satisfies 
$ [t_{k \oplus k}(f),t_{k \oplus k}(g)]_M =
\sqrt{-1}\hbar (t_k) t_{k \oplus k}(\{f,g\})
+\tilde{O}(\hbar^{1+\epsilon} (t_k))$. 
In other words,
$t_{k \oplus k}$ is a weak quantization map.
We define 
$q_{k \oplus k} : \mathcal{B}\to Y_k \oplus Y_k$
as a weak quantization map, in the same way.
Let $\pi : M_k \oplus M_k \rightarrow M_k$ 
be a projection map to one $M_k$.
From Proposition \ref{prop5},
by using a projection  
$\pi$ and commutativity $\pi \circ t_{k \oplus k}=t_k$,
we obtain $q_k = \pi \circ q_{k \oplus k}$ as a weak quantization map
and $\langle q_k (\mathcal{B})\rangle = M_k$. 
This result reproduces known facts.
\begin{align*}
\xymatrix{
\mathcal{A} \ar[r]^-{t_{k \oplus k}}\ar@/_48pt/[dddr]_{t_{k}} & 
T_k \oplus T_k \ar@{^{(}-_>}[r]
& M_k\oplus M_k \ar[dd]^{\pi}\\
\mathcal{B} \ar[r]^-{q_{k \oplus k}}\ar[rd]^{q_k} & Y_k \oplus Y_k
\ar@{^{(}-_>}[ur]
& {}\\
{}& Y_k\ar@{^{(}-_>}[r] & M_k\\
{}& T_k \ar@{^{(}-_>}[ur] & {} 
}
\end{align*}

To prevent later confusion, a note is made here.
From the fact that this projection $\pi$ is 
an $R$-algebra homomorphism but not injective one,
this $\pi$ is removed from the set of
$R$-algebra homomorphisms we will consider 
in the following parts of this paper.
As we will see in Section \ref{sect2+1},
only injective $R$-algebra homomorphisms are in $Mor(QW)$.
\end{example}

\subsection{Quantization Maps}

The parameter of noncommutativity $\hbar$ is
determined by the choice of weak quantization map
i.e. $\hbar : wkQ \rightarrow {\mathbb C}$.
The set of the weak quantization maps $wkQ$ is the
disjoint union as 
\begin{align*}
wkQ= \bigsqcup_{x\in {\mathbb R}_{\ge 0}} Qh_x ,
\end{align*}
where 
\begin{align*}
Qh_{x}:= \{ q \in wkQ \mid  | \hbar(q)| = x \in {\mathbb R}_{\ge 0} \}.
\end{align*}

We denote the set of all targets of $Qh_{x}$ by
\begin{align}
Qh_{x}( {{\mathscr Poisss}} ) := 
\{ t(q) ~| ~ q \in Qh_{x} \},
\end{align}
where 
$ t(q) $ is a target of $q$.
Particularly, $Qh_{0}:= \{ q \in wkQ \mid  | \hbar(q)| = 0 \}$ 
is a set of maps to commutative algebras including the identity
$Id_A$ for $A \in U_p ({\mathscr Poisss})$, where $U_p$ is a forgetful
functor from ${\mathscr Poisss}$ to $QW$ defined in Section \ref{sect2+1}.
\\

\begin{definition}
The set of quantization maps $Q$ is defined by
\begin{align}
Q:= wkQ \backslash Qh_{0}=\bigsqcup_{x \neq 0} Qh_x ,
\end{align}
and the set of all targets of the quantization maps is defined by
\begin{align}
Q( {{\mathscr Poisss}} ) := 
\{ t(q) ~| ~ q \in Q \}
\end{align} 
We call $q \in Q$ a quantization map.
\end{definition}
\bigskip

Here, let us explain why we define the Poisson morphism in ${\mathscr Poisss}$
is surjective.
At first, we see the following proposition.
\begin{proposition} \label{prop5+}
Let $A^\mu , A^\nu $ be Poisson algebras and let $T_i$ be an $R$-module
as a subset of an $R$-algebra $M_i$.
For a weak quantization map $t_{\nu i } : A^{\nu} \rightarrow T_i$ and a Poisson morphism $\phi : A^\mu \rightarrow A^\nu$, 
$t_{\nu i }\circ \phi : A^\mu \rightarrow T_i$ is also a weak quantization map.
\end{proposition}

\begin{proof}
\begin{align*}
\sqrt{-1}\hbar(t_{\nu i }) t_{\nu i }\circ \phi (\{f,g\})+ \tilde{O}(\hbar^{1+\epsilon})
&=\sqrt{-1}\hbar(t_{\nu i }) t_{\nu i }(\{ \phi (f), \phi(g)\})+\tilde{O}(\hbar^{1+\epsilon})\\
&= [t_{\nu i }( \phi (f)) ,~ t_{\nu i }(\phi(g))]+\tilde{O}(\hbar^{1+\epsilon}).
\end{align*}

\end{proof}

Note that
Proposition \ref{prop5+}
does not show that $t_{\nu i} \circ \phi$
is always in $Q$ even if $t_{\nu i} \in Q$.
For example, suppose that $C\subset T_i \subset M_i$ 
is a commutative subalgebra of $M_i$, and
${\rm Im}\phi \subset t_{\nu i}^{-1}(C)$.
Then the image of $t_{\nu i} \circ \phi$ is in $C$
and $\hbar (t_{\nu i} \circ \phi ) = 0$, i.e.
$t_{\nu i} \circ \phi \notin Q$.
On the other hand, if $\phi$ is surjective
$t_{\nu i} \circ \phi (A^\mu )=t_{\nu i} (A^\nu)$.
This fact and Proposition \ref{prop5+} derives the following.

\begin{proposition} \label{prop6}
Let $A^\mu , A^\nu $ be Poisson algebras.
Let $t_{\nu i } \in Q$ be a quantization map
whose source is $A^{\nu}$.
If a Poisson morphism $\phi : A^\mu \rightarrow A^\nu$ is
surjective, then
$t_{\nu i }\circ \phi \in Q$.
\end{proposition}
This property is important to make the category $QW$,
so an ancillary condition that the morphisms are surjective
is added to the conditions for $Mor ({\mathscr Poisss})$.\\
\bigskip

Next, we make sure that we can treat $\hbar $ as a continuous 
parameter in $Q$.
\begin{proposition}\label{prop7}
Let $A$ be a Poisson algebra. 
Let $T_i$ be an $R$-module.
For $\forall q_i \in Q : A \rightarrow T_i$ and  
$\forall x \in {\mathbb C}^\times$,
\begin{align}
q_i^x : = \frac{x}{\hbar({q_i})} q_i
\end{align}
is a quantization map from $A$ to $T_i$,
such that 
$\hbar(q_i^x) =x , ~ q_i(A) = q_i^x(A) $. 
\end{proposition}
\begin{proof}
From $[q_i (f) , q_i (g) ]= \sqrt{-1}\hbar (q_i) q_i(\{ f ,g \})
+\tilde{O}(\hbar^{1+\epsilon})$,
we obtain
\begin{align*}
[q_i^x (f) , q_i^x (g) ]= i x q_i^x(\{ f ,g \})
+\frac{x^2}{\hbar^2(q_i) }\tilde{O}(\hbar^{1+\epsilon}).
\end{align*}
Because $\displaystyle \frac{x^2}{\hbar^2(q_i) }\tilde{O}(\hbar^{1+\epsilon})
= \tilde{O}(x^{1+\epsilon})$, the proof is finished.
\end{proof}
This proposition teaches us that
there are uncountable infinite number 
quantizations for each target $T_i$.


\section{Category of Quantization }\label{sect2+1}

In this section, we construct a category $QW$
that describe whole space of the all quantizations
of all Poisson algebras.

We denote the target and the source of a morphism $q \in Q$
by $t(q)$ and $s(q)$, respectively.
Recall that $t(q)$ is a subspace of some
$R$-algebra for each $q$.

At the first step,
we introduce a category of target spaces $\{ t(q) |~ q \in Q \}$.

\begin{definition}
$QP$ is a subcategory of $R$-Mod defined as
follows.
\begin{itemize}
\item $ob (QP) := Q({\mathscr Poisss})= \{ t(q) |~ q \in Q \}.$
\item Suppose that two objects 
$t(q_i) , t(q_j) $ are subsets of 
$R$-algebras $M_i, M_j$, respectively.
The morphism $Mor (t(q_i), t(q_j))$ is the set of
all $h_{ij} |_{t(q_i)}: t(q_i)\rightarrow t(q_j)$. 
Here, $h_{ij} |_{t(q_i)}$
is given as an injective $R$-algebra homomorphism 
$h_{ij} : M_i \rightarrow M_j$
whose source is restricted into $t(q_i)$, and 
${\rm Im}(h_{ij}|_{t(q_i)}) \subset t(q_j)$:
\begin{align*}
&Mor (t(q_i), t(q_j)):=\\
&\{ h_{ij} |_{t(q_i)} : t(q_i)\rightarrow t(q_j)~ | ~
h_{ij} : M_i \rightarrow M_j 
\mbox{ is 
an injective $R$-algebra homomorphism } \}.
\end{align*}
\end{itemize}

\end{definition}

The composition
of $m_{ij}= h_{ij}|_{t(q_i)} \in Mor(t(q_i) , t(q_j))$ and 
$m_{jk}=h_{jk}|_{t(q_j)} \in Mor(t(q_j) , t(q_k))$ is
well-defined by
$$m_{jk} \circ m_{ij}
= h_{jk}|_{t(q_j)} \circ h_{ij}|_{t(q_i)}
=(h_{jk} \circ h_{ij})|_{t(q_i)}.
$$

Next, we consider a category 
${\mathscr Poisss} \bigsqcup QP$ defined by 
\begin{itemize}
\item $ob ( {\mathscr Poisss} \bigsqcup QP )=
ob ( {\mathscr Poisss} )  \bigsqcup  ob ( QP )$
\item $Mor ( {\mathscr Poisss} \bigsqcup QP )=
Mor ( {\mathscr Poisss} ) \bigsqcup Mor ( QP ).$
\end{itemize}

Since $ob ( {\mathscr Poisss} ) \cap  ob ( QP ) = \emptyset$
and $Mor ( {\mathscr Poisss} ) \cap  Mor ( QP ) = \emptyset$,
it is trivial that ${\mathscr Poisss} \bigsqcup QP$ is a category.

The final step to introduce 
a category $QW$
that describe whole space of the all quantizations
of all Poisson algebras is to define its morphism
by
\begin{align}\label{mor_QW}
Mor ( QW ):= Mor ( {\mathscr Poisss} ) \bigsqcup Mor ( QP ) \bigsqcup Q .
\end{align}

\begin{proposition}\label{prop8}
If $\forall f, g  \in Mor ( QW ) $ satisfy $s(f) = t(g)$,
then $f \circ g \in  Mor ( QW ) $.
\end{proposition}

\begin{proof}
For the case $ f, g  \in Mor ( {\mathscr Poisss} )$ or 
for the case $ f, g  \in Mor ( QP )$ ,
this statement is trivial.
For $\forall f, g \in Q$, $f\circ g $ never exists,
because $s(f) \in ob ({\mathscr Poisss})$ and $t(g)\in ob (QP)$.
So, we have to show this statement in the case i)
$g \in Mor ({\mathscr Poisss}), ~f\in Q$ and the case ii)
$g\in Q, ~ f \in Mor( QP )$.
For the case i), Proposition \ref{prop6} shows that 
$f \circ g \in  Q $.
For the case ii), Proposition \ref{prop3}, Corollary \ref{coro3},
and Proposition \ref{thm3}, show that $f \circ g \in  Q $, since 
every morphism in $QP$ is injective.
Therefore, this proposition is proved.
\end{proof}

Proposition \ref{prop8} guarantees  the following definition of
a category of whole space of the quantizations makes sense.

\begin{definition}[Quantization World $QW$]
The category of quantization $QW$ is a subcategory of $\mbox{$R$-Mod}$ 
whose object is defined by
\begin{align}
ob (QW) :=ob ( {\mathscr Poisss} ) \cup  ob ( QP ) 
\end{align}
and its morphism is defined by (\ref{mor_QW}).
\end{definition}

A schematic diagram of this 
category of quantization $QW$ is shown below.
\begin{align}\label{QWdiagram}
\vcenter{
\xymatrix{
\cdots \ar[r]  & M_1 \ar[r]^{ h_{12}}& M_2 \ar[r]^{ h_{23}} & M_3 \ar[r] & \cdots \\
\cdots \ar[r]  & 
\langle q_1(A) \rangle \ar[r]^{ h_{12}|_{\langle q_1(A) \rangle}}
\ar@{^{(}-_>}[u]
&
\langle q_2(A) \rangle \ar[r]^{ h_{23}|_{\langle q_2(A) \rangle}}
\ar@{^{(}-_>}[u]
&
\langle q_3(B) \rangle \ar[r]
\ar@{^{(}-_>}[u]
& \cdots \\
\cdots \ar[r]  & 
 q_1(A)  \ar[r]^{ h_{12}|_{ q_1(A) }}
\ar@{^{(}-_>}[u]
&
q_2(A) \ar[r]^{ h_{23}|_{ q_2(A) }}
\ar@{^{(}-_>}[u]
&
 q_3(B) \ar[r]
\ar@{^{(}-_>}[u]
& \cdots \\
\cdots \ar[r]  &  A \ar[r]^{ \phi_{AB}} \ar[u]^{q_1} \ar[ur]^{q_2} 
& B \ar[r] \ar[u]^{q_2'}\ar[ur]^{q_3} & \cdots & {}
}}
\end{align}
Here $M_i $ is an $R$-algebra, and $\langle q_i(A) \rangle , \cdots$ 
are  subalgebras
of $M_i , \cdots$ that generated by $q_i(A), \cdots $, respectively. $A,B,\cdots $ are Poisson algebras, and $q_i$ is a quantization map.
$\phi_{AB}$ is a surjective Poisson morphism, and $h_{ij}$ is an injective $R$-algebra homomorphism.
This figure (\ref{QWdiagram}) is not strict but just for help to
understand relations between objects in $QW$.
\\

It would not be waste to emphasize the importance of injectivity of $Mor (t(q_j), t(q_i))$ here, even though this is an overlap with 
what was previously stated around Proposition
\ref{prop3}, Corollary \ref{coro3},
and Proposition \ref{thm3}.
Consider some $q_j : A \rightarrow T_j \in Q$ i.e. 
$\hbar (q_j) \neq 0$, where 
$\langle T_j \rangle = M_j$ is a noncommutative $R$-algebra.
Suppose there exists an $R$-algebra homomorphism 
$h_{ji} : M_j \rightarrow M_i$ 
such that its image $\mbox{Im} h_{ji} \subset M_i$ is a commutative subalgebra.
For $\forall f,g \in A$,
\begin{align*}
0 &= [ h_{ji}\circ q_j (f) , h_{ji}\circ q_j (g) ] \\
&=  h_{ji} ( [ q_j (f) , q_j (g) ] ) \\
&=   \hbar(q_j)~ h_{ji} \circ q_j (\{ f , g\} ) + \tilde{O}(\hbar^{1+\epsilon} (q_j)),
\end{align*}
since $\mbox{Im} h_{ji} $ is commutative.
From $\hbar (q_j) \neq 0$, 
\begin{align*}
h_{ji} \circ q_j (\{ f , g\} )=0 , \qquad 
 \tilde{O}(\hbar^{1+\epsilon} (q_j)) =0 .
\end{align*}
Furthermore, we find that there 
exists a pair of $f, g$ such that $q_j (\{ f , g\} ) \neq 0$
since $\hbar (q_j) \neq 0$.
Then we obtain that
\begin{align}
\mbox{Ker} h_{ji}|_{T_j} \supset \{ q_j (\{ f, g \} ) ~| ~ f,g \in A \} 
\neq \{ 0 \} .
\end{align}
The presence of this kernel prevents the construction of a category 
$QW$ using $Q$.
To avoid this, injectivity is added 
to the definition for $Mor (t(q_j), t(q_i))$.


\section{Classification of Quantizations}\label{sect3}

As the classical limits of matrix regularization and so on teach us,
what characterizes quantum geometry is not only the algebra generated by
the target space of the quantization map, but also the information about the map of quantization itself.
Even if we confine our considerations to algebra, 
$\langle q_1 (A) \rangle, \langle q_2 (A) \rangle, \cdots$  are more 
important than $M_1 , M_2 , \cdots $ in the diagram (\ref{QWdiagram})
to characterize the quantization maps.
In this section, we introduce a framework that classifies quantizations.

\subsection{Equivalence of Quantization Maps}

Recall that Poisson category ${{\mathscr Poisss}}$
is a category whose objects are Poisson algebras, and 
morphisms are surjective Poisson homomorphisms.

Let us introduce $U_P$ as a forgetful functor from ${\mathscr Poisss}$ 
to $QW$ forgetting multiplication
and Poisson structure;
\begin{align*}
U_P ~ : ~ {\mathscr Poisss} \rightarrow QW .
\end{align*}
For simplicity, we abbreviate $U_P (A)$ and $ U_P(\phi )$ as $ A $
and $\phi$ for $A \in ob({\mathscr Poisss}) , ~ \phi \in Mor ({\mathscr Poisss})$
, respectively.
Next, let $I_{QP}$ be a functor from $QP$ to $QW$ such that 
all objects and morphisms are identically embedded in $QW$ as 
$I_{QP}(QP)= QP \subset QW$
i.e. $I_{QP}( T_i)=T_i \in ob(QW)$ for any object $T_i \in ob(QP)$ and 
$I_{QP}( h_{ij})= h_{ij} \in Mor(QW)$ for any morphism $h_{ij} \in Mor(QP)$.
As similar to the case of $U_P$, we abbreviate $I_{QP} (T_i)$ and $ I_{QP}(h )$ as $ T_i $
and $h$ for $T_i \in ob(QP) , ~ h \in Mor (QP)$, respectively.

Using these functors a comma category is defined as follows.
\begin{definition}[Comma category $( U_P \downarrow I_{QP}) $]
$( U_P \downarrow I_{QP})$ is a comma category whose object is a triple
\begin{align}
ob( U_P \downarrow I_{QP}) := \{
( A, q_{Ai} , T_i )~ | ~ A \in ob ({\mathscr Poisss}) ,~T_i \in ob(QP),~ 
q_{Ai}: A \rightarrow T_i \in Mor (QW) \}
\end{align}
and its morphism 
$(\phi_{AB} , t_{ij}): (A,q_{Ai},T_i) \rightarrow (B , q_{Bj} , T_j)$
is a pair of a Poisson map 
$\phi_{AB} : A \rightarrow B \in Mor ({\mathscr Poisss})$ and an $R$-module
homomorphism
$t_{ij}
: T_i \rightarrow T_j \in Mor (QP)$
such that
\begin{align}\label{comma_commute}
t_{ij} \circ q_{Ai} = q_{Bj} \circ \phi_{AB} .
\end{align}
\begin{align*}
\xymatrix{
U_P(A)=A \ar[r]^{\phi_{AB}} \ar[d]_{q_{Ai}}& B=U_P(B) \ar[d]^{q_{Bj}}  \\
I_{QP}(T_i)=T_i \ar[r]_{t_{ij}}&T_j =I_{QP}(T_j)
}
\end{align*}
\end{definition}


Using this comma category, quantization maps are classified.
For every $q_{Ai} : A \rightarrow T_i \in Q$ there is a unique 
map $q'_{Ai} : A \rightarrow q_{Ai} (A) \in Q$ defined by 
$q'_{Ai} (f ) := q_{Ai} (f)$ for any $f \in A$.
For simplicity,  $q'_{Ai}$ will simply be abbreviated as  $q_{Ai}$ in the following.
\begin{definition}
When $q_{Ai} : A \rightarrow T_i ~ \in Q$ and 
$q_{Bi} : B \rightarrow T_j ~ \in Q$ satisfy
\begin{align*}
(A , q_{Ai} , q_{Ai}(A) ) \simeq (B , q_{Bj} , q_{Bj}(B) ) 
\end{align*}
in $( U_P \downarrow I_{QP})$, that is, there exists an
isomorphism $(\phi_{AB} , t_{ij}): (A,q_{Ai},q_{Ai}(A)) \rightarrow (B , q_{Bj} , q_{Bj}(B))$, then we say that the quantization $q_{Ai}$ and $q_{Bj}$
are equivalent and we denote $q_{Ai} \simeq q_{Bj}$.
\end{definition}
This is natural definition for equivalence of quantizations.
However, the relation between algebras generating from 
$q_{Ai}(A)$ and $ q_{Bj}(B)$ are not mentioned in this definition.
Note that $q_{Ai}(A)$ and $q_{Bj}(B)$ are not algebras in general 
but just $R$-modules, 
subsets of some $R$-algebras.
So, we investigate this relation, next.\\

Let us introduce the following functor.
\begin{definition}
A functor $F : ( U_P \downarrow I_{QP}) \rightarrow {\mathscr Poisss} \times R$-alg
is defined by the map between these objects:
\begin{align}
( A, q_{Ai} , T_i ) \mapsto (A , \langle q_{Ai}(A) \rangle ),
\end{align}
and  the map between these morphisms:
\begin{align}
(\phi_{AB} , t_{ij} ) \mapsto 
( \phi_{AB} , h_{ij}|_{\langle q_{Ai}(A) \rangle} ),
\end{align}
where $h_{ij}$ is an $R$-algebra homomorphism from 
$\langle T_i \rangle \rightarrow \langle T_j \rangle$ such that
$h_{ij}|_{T_i} = t_{ij}$. 
\end{definition}

Let us make sure that this $F$ is correctly defined as a functor.
We have to check three points :
1) $ h_{ij}|_{\langle q_{Ai}(A) \rangle}$ is uniquely determined
from $t_{ij}$, 
2) $ h_{ij}( \langle q_{Ai}(A) \rangle ) \subset \langle q_{Bj}(B) \rangle$,
and
3) compositions are defined in a consistent manner.\\

1) At first, we check that
$ h_{ij}|_{\langle q_{Ai}(A) \rangle}$ is uniquely determined
from $t_{ij}$.

From $h_{ij}|_{T_i} = t_{ij}$ and $q_{Ai}(A)\subset T_i$,
$h_{ij}|_{q_{Ai}(A)}$ uniquely exists.
Therefore, we found $ h_{ij}|_{\langle q_{Ai}(A) \rangle}$ is uniquely determined from  Proposition \ref{prop1}. \\

2)Next, we would like to show $ h_{ij}( \langle q_{Ai}(A) \rangle ) \subset \langle q_{Bj}(B) \rangle$, because if the following proposition does not hold, then the composition of morphisms is not defined.

\begin{proposition}
Let $A, B$ be Poisson algebras with a surjective Poisson morphism 
$\phi_{AB} : A \rightarrow B$.
Let $q_{Ai} : A \rightarrow T_i$ and
$q_{Bj} : B \rightarrow T_j$ be quantization maps. 
Suppose that there exists $t_{ij}:T_i \to T_j$ defined 
by $t_{ij}= h_{ij}|_{T_i}$, where $h_{ij}$ is 
some $R$-algebra homomorphism
$h_{ij} : \langle T_i \rangle \to \langle T_j \rangle $.
If they satisfy the commutativity (\ref{comma_commute}),
then,
\begin{align*}
h_{ij}( \langle q_{Ai}(A) \rangle ) = \langle q_{Bj}(B) \rangle .
\end{align*}
\end{proposition}

\begin{proof}
Because $\phi_{AB}$ is surjective,
\begin{align}\label{lem_pf_1}
\langle q_{Bj}(B) \rangle = \langle q_{Bj} \circ \phi_{AB}(A) \rangle.
\end{align}
From the commutativity (\ref{comma_commute}),
\begin{align}\label{lem_pf_2}
q_{Bj} \circ \phi_{AB}= t_{ij} \circ q_{Ai} 
=  h_{ij}|_{q_{Ai}(A)} \circ q_{Ai}.
\end{align}
Proposition \ref{prop2} shows that
\begin{align}\label{lem_pf_3}
h_{ij} (\langle q_{Ai}(A) \rangle) 
=  \langle h_{ij}|_{q_{Ai}(A)} ( q_{Ai}(A)) \rangle .
\end{align}
From (\ref{lem_pf_1}),(\ref{lem_pf_2}), and (\ref{lem_pf_3}),
we obtain the conclusion that we want.
\end{proof}

By this proposition, $ h_{ij}( \langle q_{Ai}(A) \rangle ) \subset \langle q_{Bj}(B) \rangle$ is shown.\\

3) Finally, we make sure that $F$ is defined as a functor.
Consider the compositions of morphisms 
$( \phi_{AB} , t_{ij}) : ( A, q_{Ai} , T_i ) \rightarrow (B , q_{Bj} ,T_j )$  and
$( \phi_{BC} , t_{jk}): (B , q_{Bj} ,T_j ) \rightarrow (C , q_{Ck} ,T_k )$ in $( U_P \downarrow I_{QP})$.
From the results of 1) and 2)
\begin{align*}
F( \phi_{BC} , t_{jk})\circ F ( \phi_{AB} , t_{ij}) &=
( \phi_{BC} , h_{jk}|_{\langle q_{Bj}(B) \rangle}) \circ ( \phi_{AB} , h_{ij}|_{\langle q_{Ai}(A) \rangle}) \\
&= 
( \phi_{BC}\circ \phi_{AB}  , h_{jk} \circ  h_{ij}|_{\langle q_{Ai}(A) \rangle})=
F(( \phi_{BC} , t_{jk})\circ ( \phi_{AB} , t_{ij})).
\end{align*}
From the above, it is shown that $F$ is a functor.\\

$F( U_P \downarrow I_{QP})$ is a subcategory of 
${\mathscr Poisss} \times R$-alg.
We call $F( U_P \downarrow I_{QP})$ ${\mathscr Poisss}-Q({\mathscr Poisss})$ 
pair category.
In the following, this is abbreviated as $P-QP := F( U_P \downarrow I_{QP})$.
The object of $P-QP$ is a pair of a Poisson algebra and a generated algebra
from an image of a quantization map, so 
we call the object of $P-QP$ quantization pair.
Using $P-QP$, equivalence of quantization is expressed as follows.
\begin{theorem}\label{quant_equiv_thm}
Let $q_{Ai}:A \rightarrow T_i$ and $q_{Bj}: B \rightarrow T_j$ 
be quantization maps in $Q$.
In $P-QP:= F( U_P \downarrow I_{QP})$, if and only if 
there exists an isomorphism of ${\mathscr Poisss} \times R$-alg
$(\phi_{AB} , h_{ij})$
from
$(A , \langle q_{Ai}(A) \rangle )$ to $(B , \langle q_{Bj}(B) \rangle )$,
$q_{Ai}$ and $q_{Bj}$ are equivalent quantization:
\begin{align}
(A , \langle q_{Ai}(A) \rangle ) \simeq (B , \langle q_{Bj}(B) \rangle )
\ \ \ 
\Leftrightarrow 
\ \ \ q_{Ai} \simeq q_{Bj} .
\end{align}
\end{theorem}

\begin{proof}
When $(\phi_{AB} , h_{ij}) : (A , \langle q_{Ai}(A) \rangle ) \rightarrow (B , \langle q_{Bj}(B) \rangle )$ is an isomorphism in 
${\mathscr Poisss} \times R$-alg,
we put $t_{ij}= h_{ij}|_{q_{Ai}(A)}$ and $t_{ji}:=  h_{ij}^{-1}|_{q_{Bj}(B)}$.
Then for $\forall a \in q_{Ai}(A) \subset \langle q_{Ai}(A) \rangle $
\begin{align*}
t_{ji}\circ t_{ij}(a)= h_{ij}^{-1} \circ h_{ij}(a) = a,
\end{align*}
i.e. $t_{ji}\circ t_{ij}= id_{q_{Ai}(A)}$.
$t_{ij}\circ t_{ji}= id_{q_{Bj}(B)}$ is obtained similarly. 
These show that $q_{Ai} \simeq q_{Bj}$.\\
Conversely, if $q_{Ai} \simeq q_{Bj}$, there exists
an isomorphism $(\phi_{AB} , t_{ij}): (A, q_{Ai} , q_{Ai}(A)) \rightarrow (B , q_{Bj} , q_{Bj}(B))$ and
some $R$-algebra homomorphism 
$h_{ij}: \langle q_{Ai}(A) \rangle  \rightarrow \langle q_{Bj}(B) \rangle$
such that $h_{ij}|_{ q_{Ai}(A)} = t_{ij}$. 
Similarly, there exists $t_{ji}=t_{ij}^{-1}: q_{Bj}(B) \rightarrow  q_{Ai}(A)$
given as some restriction of an $R$-algebra homomorphism 
$h_{ji}$.
$\forall a \in \langle q_{Ai}(A) \rangle$ is expressed as
$\displaystyle a = \sum_k \sum_{I_k} a_{I_k} e_{i_1} \cdots e_{i_k}$
by using $e_{i_l} \in  q_{Ai}(A)$. And,
\begin{align*}
h_{ji}\circ h_{ij} (a) &=  \sum_k \sum_{I_k} a_{I_k} 
h_{ji}\circ h_{ij} (e_{i_1}) \cdots h_{ji}\circ h_{ij} (e_{i_k})
\\
&= 
\sum_k \sum_{I_k} a_{I_k} 
t_{ji}\circ t_{ij} (e_{i_1}) \cdots t_{ji}\circ t_{ij} (e_{i_k})
\\
&=  \sum_k \sum_{I_k} a_{I_k} e_{i_1} \cdots e_{i_k}= a.
\end{align*}
So, we find that $h_{ji}|_{\langle q_{Bj}(B) \rangle}\circ h_{ij}|_{\langle q_{Ai}(A) \rangle } = id_{\langle q_{Ai}(A) \rangle }$.
$h_{ij}|_{\langle q_{Ai}(A) \rangle }\circ h_{ji}|_{\langle q_{Bj}(B) \rangle} = id_{\langle q_{Bj}(B) \rangle }$ is shown in the same way. Therefore, we obtain $(A , \langle q_{Ai}(A) \rangle ) \simeq (B , \langle q_{Bj}(B) \rangle )$.
\end{proof}

\begin{rem}
We can introduce a functor $\langle - \rangle : QP \rightarrow R\mbox{-alg}$,
by $T_i \in ob (QP) \mapsto \langle T_i \rangle \in ob (R\mbox{-alg} )$ and $t_{ij} \in Mor(T_i , T_j) \mapsto h_{ij} \in Mor(\langle T_i \rangle , \langle T_j \rangle )$
where $t_{ij} = h_{ij}|_{T_i}$. 
Proposition \ref{prop1} guarantees the existence of this functor.
We denote the image of this functor by $\langle QP \rangle$.
There are natural projection functors 
$\pi_p : P-QP \rightarrow \mbox{$R$-alg}$ and 
$\pi_{\langle QP \rangle} : P-QP \rightarrow \langle QP \rangle$.
\end{rem}

\subsection{Restriction to a Single Poisson Algebra $A$}
Until now, we have considered whole space of all Poisson algebras and their
all quantizations.
However, it is also important to focus on quantizations on a single
Poisson algebra.
In this subsection, we will discuss about it.\\

Let ${\mathbf 1}$ be the trivial category with the only one object $ * $
and the only one morphism $Id_*$.
For a Poisson algebra $A$, we define a functor 
$A : {\mathbf 1} \rightarrow QW$ by
$ A (*) = A \in ob (QW)$ and $A(Id_*)= Id_{A} \in Mor (QW)$.
Let us consider the comma category associated with $A$ and $I_{QP}$.
\begin{definition}[$(A \downarrow I_{QP})$]
Let $A$ be a fixed Poisson algebra. 
$(A \downarrow I_{QP})$ is a category whose objects and morphisms are given by
\begin{align}
ob (A \downarrow I_{QP})&=
\{ (q_i, T_i):= (*, q_i, T_i)~|~ q_i : A \rightarrow T_i  \in Q , ~
T_i \in ob (QP) \},\\
Mor ((q_i, T_i), (q_j, T_j))
&= \{ t_{ij} \in Mor_{QW}(T_i , T_j) ~| ~ t_{ij} \circ q_i = q_j \}.
\end{align}
\end{definition}
Roughly speaking, $(A \downarrow I_{QP})$ is a category of all quantizations 
of $A$.
Note that $(A \downarrow I_{QP})$ is similar to a co-slice category but different.
In the same way as $P-QP$, we define the following functor.

\begin{definition}[$F_A$]
The functor $F_A : (A \downarrow I_{QP}) \rightarrow R\mbox{-alg}$ is defined by
$$F_A ( (q_i, T_i) )= \langle q_i (A) \rangle $$ 
for $\forall (q_i, T_i) \in ob (A \downarrow I_{QP})$, and
for $t_{ij} \in Mor_{ (A \downarrow I_{QP})}((q_i, T_i),(q_j, T_j))$
$$ F_A( t_{ij} )= h_{ij} |_{\langle q_i (A) \rangle},$$
where $h_{ij} $ is an $R$-algebra homomorphism satisfying 
$h_{ij}|_{T_i}= t_{ij}$.
\end{definition}

From Theorem \ref{quant_equiv_thm} with $A=B$, we obtain the following.

\begin{corollary}\label{coro_4.8}
Let $q_i$ and $q_j$ be quantization maps whose source is a Poisson algebra $A$.
$q_i$ and $q_j$ are equivalent quantization of $A$ in $F_A (A \downarrow I_{QP})$
if and only if object $\langle q_i (A) \rangle \in R\mbox{-alg}$  
is isomorphic 
to $\langle q_j (A) \rangle \in R\mbox{-alg}$
in $R\mbox{-alg}$.
\end{corollary}
From this corollary, we found that 
$F_A (A \downarrow I_{QP})$ classifies quantization of $A$.\\

\begin{example}\label{ex_poly}
Let ${\mathcal P}({\mathbb C}^2)$ be the algebra of polynomials 
of a coordinate $(x,y)\in {\mathbb C}^2$.
By introducing a Poisson bracket by
$$\{ f , g \} = \partial_x f \partial_y g - \partial_y f \partial_x g ,$$
for any $f,g \in {\mathcal P}({\mathbb C}^2)$,
$({\mathcal P}({\mathbb C}^2), \cdot ,\{~ ,~ \} )$ is 
regarded as a Poisson algebra.

We use the notation $x^1:= x,~ x^2 = y ,~ \partial_1 = \dfrac{\partial}{\partial x},~\partial_2 = \dfrac{\partial}{\partial y}$.
Let 
$
H = (H_{i j} )  
$
be a $2\times 2$ complex matrix.
We introduce a bi-differential operator ${\overleftrightarrow H}$ as
\begin{align}
f {\overleftrightarrow H}g = H_{ij} \partial_i f \partial_j g,
\qquad
( f, g \in {\mathcal P}({\mathbb C}^2)) .
\label{bi-diff-1}
\end{align}
Here $H_{ij} \partial_i f \partial_j g$ is an abbreviation for 
$ \sum_{i,j=1}^n H_{ij} \partial_i f \partial_j g$.
Moyal product $*_H$ is defined by
\begin{align*}
f *_H g &=  f ( \exp \nu \overleftrightarrow H ) g 
= \sum_{l=0}^\infty \frac{\nu^l}{l !}H_{i_1 j_1} \cdots H_{i_l j_l}
(\partial_{i_1} \cdots \partial_{i_l} f )\cdot
(\partial_{j_1} \cdots \partial_{j_l} g )
\end{align*}
where $\nu$ is a complex number.
Then, by the Moyal product $*_{H}$, the strict deformation quantization 
$({\mathcal P}({\mathbb C}^2), *_{H} )$ is a ${\mathbb C}$-algebra \cite{OMMY}.
We define $q_H : ({\mathcal P}({\mathbb C}^2), \cdot ,\{~ ,~ \} ) \rightarrow ({\mathcal P}({\mathbb C}^2), *_{H} )$ by the inclusion map 
i.e. $q_{H} (\sum_{ij} a_{ij} x^i y^j )= \sum_{ij} a_{ij} x^i y^j $.
It is clear that $q_H$ is a quantization map.

Let us decompose $H=(H_{ij})$
into a symmetric matrix and a skew symmetric matrix as
\begin{align*}
&H= J + K  , ~ J= (J_{ij} ) ,  K = (K_{ij}) , \mbox{where } 
J_{ij}= - J_{ji} , ~ K_{ij} = K_{ji} .
\end{align*}
Here we introduce a differential operator
\begin{align}
T= \exp(-\frac{1}{2} \nu K_{ij} \partial_i \partial_j ) 
= \sum_{l=0}^{\infty} \frac{1}{l!}
(-\frac{1}{2} \nu K_{ij} \partial_i \partial_j )^l \label{1_12}
\end{align}
that act a polynomial $f$ as
\begin{align}
T({f}) = f+
\sum_{l=1}^{\infty} \nu^l 
\frac{1}{l!}
(-\frac{1}{2} K_{m j} \partial_m \partial_j )^l  f .
\end{align}

To compare two algebra $({\mathcal P}({\mathbb C}^2), *_{J} )$
$({\mathcal P}({\mathbb C}^2), *_{H} )$,
the following theorem is useful.
\begin{theorem}[Tomihisa-Yoshioka \cite{TomishitaYoshioka}]\label{inter}
$({\mathcal P}({\mathbb C}^2), *_{H} )$ and 
$({\mathcal P}({\mathbb C}^2), *_{J} )$ are ${\mathbb C}$-algebra isomorphic.
$T$ defined by (\ref{1_12}) satisfies
\begin{enumerate}
\item $T : ({\mathcal P}({\mathbb C}^2), *_{H} )\rightarrow 
({\mathcal P}({\mathbb C}^2), *_{J} )$ is linear and bijective.
\item $T(1)=1$
\item For $f, g \in ({\mathcal P}({\mathbb C}^2), *_{H} )$,
$T(\tilde{f} *_H \tilde{g} ) = 
T(\tilde{f}) *_J T(\tilde{g} ) $.
\end{enumerate}
In other words, $T$ is an $R$-algebra isomorphism. 
\end{theorem}

We define 
$q_J : ({\mathcal P}({\mathbb C}^2), \cdot ,\{~ ,~ \} ) \rightarrow 
({\mathcal P}({\mathbb C}^2), *_{J} )$ as an inclusion map 
i.e. $q_{J} (\sum_{ij} a_{ij} x^i y^j )= \sum_{ij} a_{ij} x^i y^j $, too.
Then
$\langle q_H (  {\mathcal P}({\mathbb C}^2) ) \rangle =  ({\mathcal P}({\mathbb C}^2), *_{H} ) \simeq
\langle q_J (  {\mathcal P}({\mathbb C}^2) ) \rangle = ({\mathcal P}({\mathbb C}^2), *_{J} )$.
From Corollary \ref{coro_4.8}, 
$q_H$ and $q_J$
are equivalent quantizations.
\end{example}


\section{Classical Limit}\label{sect_ClassicalLimit}
One of the most important problem is :
how can we find the classical limit from the
spaces given by the quantizations of Poisson algebras?
To make the point of this problem easier to understand, 
let us consider 
matrix regularization of $S^2$ and $T^2$.
(In Appendix \ref{appendix_fuzzy},
we summarize the matrix regularization of $S^2$ and $T^2$.)
As an well-known fact,
both images of quantization maps from ${\mathcal A}$
and ${\mathcal B}$ generate the same matrix algebras.
This example shows that the classical limit is not determined 
by observation of the algebra obtained from a single quantization.
This shows that the Poisson algebra cannot be determined as the classical limit from the algebra obtained by quantization. 
But this example also teaches 
that the classical limit may be identified from the sequence of 
modules before generating algebras.\\

Since a new formulation of quantization using category 
theory was introduced in Section \ref{sect2+1}, 
we can define various classical limits in the framework.
The purpose of this section is to define some 
classical limits, study their properties 
by examining them with concrete examples, 
and look for useful definitions of classical limits 
for approaching inverse problems.
For this purpose, we pay particular attention 
to the way the sequence of $R$-modules obtained 
by the quantization maps 
determines the Poisson algebra.
\\

In the following $\displaystyle \lim_{\leftarrow {\mathcal I}} D$
means the categorical limit 
in this article.
For example we can see its definition in Chapter $V$ 
in \cite{limit2} or Chapter $5$ in \cite{limit1}.

\subsection{Naive Classical Limit}\label{naive_classical_limit}
For a start, the classical limit is defined as follows using a naive definition of the limit (in the sense of category theory).

\begin{definition}[Naive classical limit of $D({\mathcal I})$]
Let ${\mathcal I}$ be an index category, and
let $D : {\mathcal I} \rightarrow QW$ be a
diagram of shape ${\mathcal I}$.
Let $A$ be a Poisson algebra.  
Suppose that ${\mathbf t}_A$ is a set of 
$ t_i : A \rightarrow D(i) \in Q~ ( i \in {\mathcal I} )$ 
such that 
$t_i$ and $t_j$ in ${\mathbf t}_A$ satisfy
$ t_j = D(u_{ij})\circ t_i $ for 
$\forall u_{ij} \in Mor_{\mathcal I}(i,j) $, 
that is $(A, {\mathbf t}_A)$ is a cone over $D$.
When ${\mathcal I}$ and $D$ have a limit $(A, {\mathbf t}_A)$:
\begin{align}
\lim_{\leftarrow {\mathcal I}} D = 
(A \xrightarrow{t_i} D(i))_{i \in {\mathcal I}},
\end{align}
and $\displaystyle \lim_{\leftarrow {\mathcal I}} D $ is called the
naive classical limit of $D({\mathcal I})$.
\end{definition}
This definition may lead some readers to think that the naive classical limit is unrelated to the conventional classical limit. However, the following theorem shows that the naive classical limit is a generalization that includes the classical limit when only one type of quantization for one Poisson algebra is considered conventionally.

\begin{theorem}\label{thm_5_1_limit}
Let ${\mathcal I}$ be an index category, and
let $D : {\mathcal I} \rightarrow QW$ be a
diagram of shape ${\mathcal I}$
such that there is only one Poisson algebra $A$ in $D({\mathcal I})$.
If $(A, {\mathbf t}_A)$ is a cone over $D$ 
such that $Id_A \in {\mathbf t}_A$, 
then
$(A, {\mathbf t}_A)$ is the naive classical limit.
\end{theorem}

\begin{proof}
In the following, we denote $( A, \cdot ,\{~ ,~ \} ) $ by $ A$
for simplicity, 
in contexts where there is no risk of confusion.
The cone $( A , {\mathbf t})$ is expressed as:
\begin{align*}
\vcenter{
\xymatrix{
( A, \{,\})\ar[d]^{Id_{ A}} \ar[rd]_{t_{1}} \ar[rrd]^{t_{2}\cdots}&{}& {}&\\
( A, \{,\}) \ar[r]_{t_1}& 
( T_1, *_1) \ar[r]^{t_{12}} &
  ( T_2, *_2) \ar[r] &  \cdots
}}.
\end{align*}
Here each $t_i : A \rightarrow T_i$ is a quantization map,
and $t_{ij}=D(u_{ij}) : T_i \to T_j$ is in $Mor_{QP}(T_i , T_j)$.
$D({\mathcal I})$ is expressed as
\begin{align*}
\vcenter{
\xymatrix{
( A, \{,\}) \ar[r]^{t_1}& 
( T_1, *_1) \ar[r]^{t_{12}} &
  ( T_2, *_2) \ar[r] &  \cdots
}}.
\end{align*}

If there is another cone of $(B , {\mathbf q}')$ , $B$ is a Poisson algebra, since
$ A$ is a Poisson algebra, and only Poisson algebras
have possibility to be sources of morphisms in $QW$
which target is $A$. 
We denote this Poisson morphisms by $\phi : B \rightarrow  A$.
${\mathbf q}'= \{\phi , q_{B1}' , q_{B2}', \cdots \}$ is a set of morphisms that satisfy the following commutative diagram:
\begin{align*}
\vcenter{
\xymatrix{
B\ar[d]^{\phi} \ar[rd]_{q'_{B1}} \ar[rrd]^{q'_{B2}\cdots}&{}& {}&\\
( A, \{,\}) 
\ar[r]_{t_1}& 
( T_1, *_1) \ar[r]^{t_{12}} &
  ( T_2, *_2) \ar[r] & \ar@<0.5ex>[l] \cdots
}}.
\end{align*}

Suppose two $\phi, \phi' : B \rightarrow  A$
satisfy the following commutative diagram:
\begin{align*}
\vcenter{
\xymatrix{
( A, \{,\})\ar[d]_{Id_{ A}} \ar[rd] \ar[rrd]& \ar@<0.5ex>[l]^-{\phi'} \ar[l]_-{\phi} B \ar[dl]^-{\phi} \ar[d]^{q_{B1}'} \ar[dr]^{q_{B2}' \cdots}& {}&\\
( A, \{,\}) \ar[r]_{q_p}& 
( A, *_p) \ar[r]^{t_{pr}} &
 \ar@<0.5ex>[l]^{t_{rp}} ( A, *_r) \ar[r] & \ar@<0.5ex>[l] \cdots
}}.
\end{align*}
From the condition that 
$( A , \{,\})$ is an object in $D({\mathcal I})$,
$Id_{ A}$ is also a morphism in $D({\mathcal I})$
because $D$ is a functor.
From the commutativity $\phi = Id_{ A}\circ \phi' =\phi'$,
we obtain the uniqueness of $\phi$.
\end{proof}
In practice, the condition $Id_A \in {\mathbf t}_A$
in this theory can be relaxed. 
Theorem \ref{thm_5_1_limit} holds 
even if $Id_A$ is replaced by an arbitrary automorphism.

In the following subsections, we will actually consider two examples of 
quantization families corresponding to only one type of quantization 
for a single Poisson algebra and check that their naive limits are in fact what is expected.

\subsection{Example of Naive Classical Limit  ($ {\mathcal P}({\mathbb C}^2)$)}\label{ex_naive_lim_pC2}

We consider ${\mathcal P}({\mathbb C}^2)$ 
, the set of all polynomials of $x$ and $y$, and its quantization
by deformation quantization, again. 
This is an example of fixing the commutative ring $R$ to ${\mathbb C}$.
In the following, we denote $( {\mathcal P}({\mathbb C}^2), \cdot ,\{~ ,~ \} ) $ by $ {\mathcal P}({\mathbb C}^2)$
for simplicity, 
in contexts where there is no risk of confusion.
(See Example \ref{ex_poly}.)
 \\

Let introduce sequence of Moyal products 
with $p   ~(p \in {\mathbb Q})$ by
\begin{align*}
f *_p g = f \exp \left( 
\overleftarrow{\partial}_x \frac{i p}{2} \overrightarrow{\partial}_y
-\overleftarrow{\partial}_y \frac{i p}{2} \overrightarrow{\partial}_x
\right) g ,
\end{align*}
where $f,g \in  {\mathcal P}({\mathbb C}^2)$.
We define $q_\hbar : ( {\mathcal P}({\mathbb C}^2), \cdot ,\{~ ,~ \} ) \rightarrow ( {\mathcal P}({\mathbb C}^2), *_{\hbar} )$ by inclusion map 
i.e. $q_{\hbar} ( f )= f$, as similar to Example \ref{ex_poly}.

Let ${\mathcal I}_{moyal}$ be an index category 
defined by  $ob ( {\mathcal I}_{moyal}) = {\mathbb Q}$
and 
\begin{align*}
Mor (p,r) &:= \{( p \mapsto r  =: +r -p )\} ~\mbox{for} ~ r\neq0 ,p \in {\mathbb Q} ,\\
Mor (p,0) &:= \emptyset .
\end{align*}
Note that $+p -p$ is an identity for $p$.

Let  $D_{moyal} : {\mathcal I}_{moyal} \rightarrow QW$
be a diagram of shape ${\mathcal I}_{moyal}$
defined by
\begin{align*}
D_{moyal}(p) &= ( {\mathcal P}({\mathbb C}^2), *_{p } )~ 
\mbox{for }~ \forall  p\neq 0 \in {\mathbb Q} \\
D_{moyal}(0) &= ( {\mathcal P}({\mathbb C}^2), \cdot , \{ , \} )~ 
\mbox{for }~  0 \in {\mathbb Q}
\end{align*}
and
\begin{align*}
D_{moyal}( +r-p ) &= t_{p,r} \\
D_{moyal}( +r-0 ) &=  q_{r}.
\end{align*}
Here $t_{p,r}$ is an intertwiner between 
$( {\mathcal P}({\mathbb C}^2), *_{p } )$ and $( {\mathcal P}({\mathbb C}^2), *_{r } )$.
The intertwiner is given as a ${\mathbb C}$-algebra isomorphism.
The existence of such intertwiner is shown in 
\cite{bayen1,bayen2,bertelson1,bertelson2}. 
The existence of this $t_{p,r}$ is not essential 
to the following discussion. 
We choose here each isomorphism $t_{p,r}$ 
satisfying $t_{p,r}\circ q_p = q_r$.

The diagram $D_{moyal}({\mathcal I}_{moyal})$ is given as:
\begin{align*}
\xymatrix{
( {\mathcal P}({\mathbb C}^2), \{,\}) \ar[r]_{q_p}& 
( {\mathcal P}({\mathbb C}^2), *_p) \ar[r]^{t_{pr}} &
 \ar@<0.5ex>[l]^{t_{rp}} ( {\mathcal P}({\mathbb C}^2), *_r) \ar[r] & \ar@<0.5ex>[l] \cdots
}.
\end{align*}

Then, 
$( {\mathcal P}({\mathbb C}^2) , {\mathbf q})$ is a cone:
\begin{align*}
\vcenter{
\xymatrix{
( {\mathcal P}({\mathbb C}^2), \{,\})\ar[d]^{Id_{ {\mathcal P}({\mathbb C}^2)}} \ar[rd]_{q_{p}} \ar[rrd]^{q_{r}\cdots}&{}& {}&\\
( {\mathcal P}({\mathbb C}^2), \{,\}) \ar[r]_{q_p}& 
( {\mathcal P}({\mathbb C}^2), *_p) \ar[r]^{t_{pr}} &
 \ar@<0.5ex>[l]^{t_{rp}} ( {\mathcal P}({\mathbb C}^2), *_r) \ar[r] & \ar@<0.5ex>[l] \cdots
}}.
\end{align*}
From Theorem \ref{thm_5_1_limit}, 
$( ( {\mathcal P}({\mathbb C}^2), \cdot ,\{~ ,~ \} ) , {\mathbf q})$ is the naive classical limit.

\begin{corollary}
Let ${\mathbf q}$ be a set of the above quantization 
$q_p : ( {\mathcal P}({\mathbb C}^2), \cdot ,\{~ ,~ \} ) \rightarrow ( {\mathcal P}({\mathbb C}^2), *_p ) ~(p\in {\mathbb Q})$
and 
identity $Id_{ {\mathcal P}({\mathbb C}^2)}$ for Poisson algebra 
$( {\mathcal P}({\mathbb C}^2), \cdot ,\{~ ,~ \} ) $
 i.e.
${\mathbf q}= \{ q_p \in Q ~|~ p\in {\mathbb Q}\}\cup \{ Id_{ {\mathcal P}({\mathbb C}^2)}\}$. 
For the above ${\mathcal I}_{moyal}$, $D_{moyal}$,
$( ( {\mathcal P}({\mathbb C}^2), \cdot ,\{~ ,~ \} ) , {\mathbf q})$ is the naive classical limit.
\end{corollary}

\subsection{Example of Naive Classical Limit (Fuzzy sphere)}\label{naive_classical_lim_FzS2}

In this subsection, let us consider the
naive classical limit of fuzzy spheres.
Basic knowledge for the fuzzy sphere 
and the notations used in this subsection 
are summarized in the Appendix \ref{appendix_fuzzy}.\\

Let $P[x^a]$ be the algebra of polynomial generated by $x^a$. 
For an ideal $I$ which is generated by the relation $(\ref{sphere})$, we consider a Poisson algebra $\mathcal{A}=P[x^a]\slash I$ with Poisson bracket
$
\{x^a,x^b\}=\epsilon^{abc}x^c. 
$
For arbitrary $f\in \mathcal{A}$ is given as
\begin{align*}
f=f_0+f_ax^a+\frac{1}{2}f_{ab}x^ax^b+\cdots
\end{align*}
where $f_{a_1\cdots a_i}\in \mathbb{C}$ is completely symmetric and trace-free. In the case of $k\ge 2$, quantization maps $t_k:\mathcal{A}\to T_k := t_k (\mathcal{A})\subset M_k={\rm Mat}_{k}$ are defined by
\begin{align*}
t_k(f)&:=f_0{\bf 1}_k+f_{a_1}t_k(x^{a_1})+\cdots +
\frac{1}{(k-1)!}
f_{a_1\cdots a_{k-1}}t_k(x^{a_1}\cdots x^{a_{k-1}})\\
t_{k}(x^a)&:=\hbar_k J^a_{k}=X_k^a,
\end{align*}
where $J_a (a=1,2,3)$ satisfy $[J^a_k,J^b_k]=i\epsilon^{abc}J^c_k$, that is
 $J_a$ are generators of $k$-dimensional irreducible representation of $\mathfrak{su}(2)$.
Its Casimir relation
$
J^a_kJ_{ka}=\frac{1}{4}(k^2-1){\bf 1}_k
$
and $(\ref{sphere})$ derive the relation
$
\dfrac{4}{k^2-1}=\hbar^2_k.
$
The commutation relation of $J^a_k$
or
(\ref{eq.com}) shows that $t_k \in Q $ $(k \ge 2 )$. 

Let  ${\mathcal I}_{fzS2}$ be an index category 
defined by  $ob ( {\mathcal I}_{fzS2}) = {\mathbb Z}_{> 0}$
and 
\begin{align*}
Mor (1,k) &:=\{( 1 \mapsto k  )=:\hat{k} \} ~\mbox{for } ~ k \in {\mathbb Z}_{> 0} ,\\
Mor (p,q) &:= \emptyset  ~\mbox{for integers }~ p,q >1 ,~ p \neq q ,\\
Mor (p,p) &:= Id_p.
\end{align*}
The absence of a morphism between $p$ and $q$
($p,q >1$ and $p\neq q$) in the index category is due to the fact that there is no morphism between $T_p $ and $T_q$.
Since $\ker t_p = \{ f_{a_1 \cdots a_p} x^{a_1} \cdots x^{a_p} +  f_{a_1 \cdots a_{p+1}} x^{a_1} \cdots x^{a_{p+1}} + \cdots \}$,
there is no injective morphism between $T_p $ and $T_q$ that satisfies the commutative diagram with $t_p$ and $t_q$.

Let  $D_{fzS2} : {\mathcal I}_{fzS2} \rightarrow QW$
be a diagram of shape ${\mathcal I}_{fzS2}$
defined by
\begin{align*}
D_{fzS2}(1) &= \mathcal{A}~ 
\mbox{for}~  1 ,\\
D_{fzS2}(k) &= T_k~ 
\mbox{for}~ \forall  k\neq 1 \in {\mathbb Z}_{> 0} ,
\end{align*}
and
\begin{align*}
D_{fzS2}( \hat{1} ) &= Id_{\mathcal{A}} ,\\
D_{fzS2}( \hat{k} ) &=  t_{k} ~\mbox{for}~ \forall  k > 1 ,\\
D_{fzS2}( Id_k ) &= Id_{T_k} ~\mbox{for}~ \forall  k > 1 ,
\end{align*}
\begin{align*}
\left(
\vcenter{
\xymatrix{
1 \ar[d]^{\hat{2}} \ar[rd]^{\!\!\!\hat{3}} \ar[rrd]^{\hat{4} ~~\cdots}& {}&\\
2& 
3&
4~~ \cdots
}}\right)
\xrightarrow{D_{fzS2}} 
\left(
\vcenter{
\xymatrix{
\mathcal{A} \ar[d]^{t_2} \ar[rd]^{\!\!\! t_3} \ar[rrd]^{t_4 ~~\cdots}& {}&\\
T_2& 
T_3&
T_4~~ \cdots
}}\right).
\end{align*}

For the ${\mathcal I}_{fzS2}$ and $D_{fzS2}$,
the naive classical limit is determined. 
\begin{corollary}
Let ${\mathbf t}$ be a set of the above quantizations 
${\mathbf t} = \{ t_k \in Q ~|~ k >1 \}
\cup \{ Id_{\mathcal{A}} \}
$.
For the above ${\mathcal I}_{fzS2}$, $D_{fzS2}$,
$( \mathcal{A} , {\mathbf t})$ is the naive classical limit.
\end{corollary}
This follows trivially from Theorem \ref{thm_5_1_limit}.\\

\subsection{Weak Classical Limit}\label{subsec_weak_limit}
As we have seen, in the naive classical limit, 
the limit is trivially determined for some kinds of diagrams
$D({\mathcal I})$. 
From the viewpoint of discussing inverse problems, 
the range of the naive classical limit is too wide. 
In this subsection, we give a definition of the other classical limit 
with inverse problems in mind.
We introduce it by restricting the diagram to $QP$ and by using a limit that differs from the limit of category theory.
\\
\bigskip

Let us introduce a quantization family and a kind of classical limit.
\begin{definition}[Quantization Family of $A$, weak classical limit of $D({\mathcal I})$]
Let ${\mathcal I}$ be an index category, and
let $D : {\mathcal I} \rightarrow QP(=I_{QP}(QP)) \subset QW$ be a
diagram of shape ${\mathcal I}$.
Let $A$ be a Poisson algebra.  
Suppose that ${\mathbf t}_A$ is a set of 
$ t_i : A \rightarrow D(i) \in Q~ ( i \in {\mathcal I} )$ 
such that 
$t_i$ and $t_j$ in ${\mathbf t}_A$ satisfy
$ t_j = D(u_{ij})\circ t_i $ for 
$\forall u_{ij} \in Mor_{\mathcal I}(i,j) $, 
that is, $(A, {\mathbf t}_A)$ is a cone over $D$.
$D({\mathcal I})$ is called a quantization family of $A$, 
and $(A, {\mathbf t}_A)$ is called the
weak classical limit of $D({\mathcal I})$,
when the cone $(A, {\mathbf t}_A)$ satisfies the fillowing condition:
\begin{itemize}
\item For any other cone  $(B, {\mathbf q}_B)$ over $D$, 
where $B \in ob ({\mathscr Poisss}) \subset ob(QW)$, a unique map $\phi_{BA} : B \to A \in Mor_{QW}(B,A) $
such that $q_i : B \to D(i) \in {\mathbf q}_B$ satisfies
\begin{align*}
q_i = t_i \circ \phi_{BA} 
\end{align*}
for all $i \in {\mathcal I}$.
\end{itemize}
\begin{align*} 
\xymatrix{
{A} \ar[d] \ar[rd] \ar[rrd] & 
\ar@{.>}[l]_{\phi_{BA}}
B
\ar[dl] \ar[d] \ar[rd]^{{}~~~\cdots}
& &:{\mathscr Poisss}(=U_p ({\mathscr Poisss})) \subset QW \\
D(i)& 
D(j)&
D(k)\cdots &:QP(=I_{QP}(QP)) \subset QW
} 
\end{align*}
\end{definition}

\subsection{Examples of Candidates of Weak Classical limits }\label{ex_weak_lim}

We consider ${\mathcal P}({\mathbb C}^2)$ 
, again.
As in Subsection \ref{ex_naive_lim_pC2}, 
the sequence of Moyal products 
with $p   ~(p \in {\mathbb Q}^\times)$ by
$
f *_p g = f \exp \left( 
\overleftarrow{\partial}_x \frac{i p}{2} \overrightarrow{\partial}_y
-\overleftarrow{\partial}_y \frac{i p}{2} \overrightarrow{\partial}_x
\right) g ,
$
where $f,g \in  {\mathcal P}({\mathbb C}^2)$, and
$q_\hbar : ( {\mathcal P}({\mathbb C}^2), \cdot ,\{~ ,~ \} ) \rightarrow ( {\mathcal P}({\mathbb C}^2), *_{\hbar} )$ is defined by $q_{\hbar} ( f )= f$.

${\mathcal I}_{moyal}'$ differs from the one in  Subsection \ref{ex_naive_lim_pC2} in that $0$ is excluded from the objects in ${\mathcal I}_{moyal}$.
Let ${\mathcal I}_{moyal}'$ be an index category 
defined by  $ob ( {\mathcal I}_{moyal}') = {\mathbb Q}^\times$
and 
\begin{align*}
Mor (p,r) &:= \{( p \mapsto r  =: +r -p )\} ~\mbox{for} ~ r,p \in {\mathbb Q}^\times.
\end{align*}
Note that $+p -p$ is an identity for $p$.

Let  $D_{moyal}' : {\mathcal I}_{moyal}' \rightarrow QP \subset QW$
be a diagram of shape ${\mathcal I}_{moyal}'$
defined by
$
D_{moyal}'(p) = ( {\mathcal P}({\mathbb C}^2), *_{p } )~ 
\mbox{for }~ \forall  p \in {\mathbb Q}^\times 
$
and
$
D_{moyal}'( +r-p ) = t_{p,r} .
$
Here $t_{p,r}$ is an intertwiner as the same one in 
Subsection \ref{ex_naive_lim_pC2}.

The diagram $D_{moyal}'({\mathcal I}_{moyal}')$ is given as:
\begin{align*}
\xymatrix{
({\mathcal P}({\mathbb C}^2), *_p) \ar[r]^{t_{pr}} &
 \ar@<0.5ex>[l]^{t_{rp}} ({\mathcal P}({\mathbb C}^2), *_r) \ar[r] & 
\ar@<0.5ex>[l] \cdots
}.
\end{align*}

Let ${\mathbf q}'$ be a set of the above quantization 
$q_p : ( {\mathcal P}({\mathbb C}^2), \cdot ,\{~ ,~ \} ) \rightarrow ( {\mathcal P}({\mathbb C}^2), *_p ) ~(p\in {\mathbb Q}^\times)$.
Even with these change, 
$( {\mathcal P}({\mathbb C}^2) , {\mathbf q}')$ is still a cone:
\begin{align*}
\vcenter{
\xymatrix{
( {\mathcal P}({\mathbb C}^2), \{,\})\ar[d]_{q_{p}} \ar[rd]^{q_{r}\cdots}&{}&\\
( {\mathcal P}({\mathbb C}^2), *_p) \ar[r]^{t_{pr}} &
 \ar@<0.5ex>[l]^{t_{rp}} ( {\mathcal P}({\mathbb C}^2), *_r) \ar[r] & \ar@<0.5ex>[l] \cdots
}}.
\end{align*}

At this case, the following proposition holds.
\begin{proposition}
If there exist the weak classical limit 
for the above ${\mathcal I}_{moyal}'$, $D_{moyal}'$,
the weak classical limit is isomorphic to
$( ( {\mathcal P}({\mathbb C}^2), \cdot ,\{~ ,~ \} ) , {\mathbf q'})$.
\end{proposition}

\begin{proof}
Let $(L, {\mathbf q}_{lim}) $ be the weak classical limit.
Then there exists a unique Poisson morphism 
$\phi_{PL} : {\mathcal P}({\mathbb C}^2) \to L$
such that $q_p = q_{Lp}\circ \phi_{PL}$ for any 
$q_{Lp}: L \to ( {\mathcal P}({\mathbb C}^2), *_p) \in {\mathbf q}_{lim} $ and any $q_p \in {\mathbf q}'$.
Recall that this  $\phi_{PL}$ is a surjective Poisson morphism.
By definition of $q_p$,~ $f= q_p(f) =  q_{Lp}\circ \phi_{PL} (f)$ for 
$\forall f \in {\mathcal P}({\mathbb C}^2)$.
Then $\phi_{PL}$ is injective. 
\end{proof}

In fact, if there is another cone, 
from which there is a morphism to ${\mathcal P}({\mathbb C}^2)$
satisfying commutativity, then it is the uniqe one.

\begin{proposition}\label{poly_unique_mor}
Let $( ( {\mathcal P}({\mathbb C}^2), \cdot ,\{~ ,~ \} ) , {\mathbf q'})$
be a cone defined above.
Let $(B, {\mathbf t})$ be the other cone
where $B$ is arbitrary object in $QW$.
If there is $\phi : B \to  {\mathcal P}({\mathbb C}^2) \in Mor (QW)$ satisfies the following 
commutative diagram:
\begin{align*} 
\vcenter{
\xymatrix{
{\mathcal P}({\mathbb C}^2) \ar[d]_{q_{p}} \ar[rd] & {}
\ar[l]_{\phi}
B
\ar[dl]_-{t_p} \ar[d]^{t_r~~~\cdots} 
&\\
( {\mathcal P}({\mathbb C}^2), *_p) \ar[r]^{t_{pr}} &
 \ar@<0.5ex>[l]^{t_{rp}} ( {\mathcal P}({\mathbb C}^2), *_r) \ar[r] & \ar@<0.5ex>[l] \cdots
}} ,
\end{align*}
then $\phi$ is unique.
\end{proposition}

\begin{proof}
Suppose two different 
$\phi, \phi' : B \rightarrow  {\mathcal P}({\mathbb C}^2)$
satisfy the following commutative diagram:
\begin{align*} 
\vcenter{
\xymatrix{
{\mathcal P}({\mathbb C}^2) \ar[d]_{q_{p}} \ar[rd] & {}
\ar[l]_{\phi, \phi'}
B
\ar[dl]_-{t_p} \ar[d]^{t_r~~~\cdots} 
&\\
( {\mathcal P}({\mathbb C}^2), *_p) \ar[r]^{t_{pr}} &
 \ar@<0.5ex>[l]^{t_{rp}} ( {\mathcal P}({\mathbb C}^2), *_r) \ar[r] & \ar@<0.5ex>[l] \cdots
}} .
\end{align*}
There are $b \in B$ such that $\phi (b) \neq \phi' (b)$.
From the commutativity $t_p = q_p\circ \phi= q_p\circ \phi'$.
By definition of each $q_p$,
$ \phi (b) = q_p\circ \phi (b) = q_p\circ \phi'(b) = \phi'(b)$.
This is a contradiction.
\end{proof}

Unfortunately, we do not know if the weak classical limit actually exists
for this ${\mathcal I}_{moyal}'$, $D_{moyal}'$.
\\
\bigskip

Next, as similar in Subsection \ref{naive_classical_lim_FzS2},
let us consider the example of fuzzy sphere, again.
Let $(\mathcal{A}=P[x^a]\slash I , ~\{ , \}) $
be a Poisson algebra, where
$
\{x^a,x^b\}=\epsilon^{abc}x^c ~ (a,b,c =1,2,3). 
$
$f\in \mathcal{A}$ is given as
$
f=f_0+f_ax^a+\frac{1}{2}f_{ab}x^ax^b+\cdots
$
where $f_{a_1\cdots a_i}\in \mathbb{C}$ is completely symmetric and trace-free. Quantization maps 
$t_k:\mathcal{A}\to T_k := t_k (\mathcal{A})\subset M_k={\rm Mat}_{k}$ are defined by
$
t_k(f):=f_0{\bf 1}_k+f_{a_1}X^{a_1}_{k}+\cdots +
\frac{1}{(k-1)!}
f_{a_1\cdots a_{k-1}}
X^a_{k} \cdots X^{a_{k-1}}_k$.

Let  ${\mathcal I}_{fzS2}'$ be an index category 
defined by  $ob ( {\mathcal I}_{fzS2}) = \{2,3,4, \cdots \}$
and 
\begin{align*}
Mor (p,q) &:= \emptyset  ~\mbox{for integers }~ p,q >1 ,~ p \neq q ,\\
Mor (p,p) &:= Id_p.
\end{align*}
In other words, ${\mathcal I}_{fzS2}'$ is a discrete category.
Let  $D_{fzS2}' : {\mathcal I}_{fzS2}' \rightarrow QP \subset QW$
be a diagram of shape ${\mathcal I}_{fzS2}'$
defined by
\begin{align*}
D_{fzS2}'(k) &= T_k~ 
\mbox{for}~ \forall  k\in \{2,3,4, \cdots \},
\end{align*}
and
\begin{align*}
D_{fzS2}'( Id_k ) &= Id_{T_k} ~\mbox{for}~ \forall k \in \{2,3,4, \cdots \} .
\end{align*}
\begin{align*}
\left(
\vcenter{
\xymatrix{
2& 
3&
4~~ \cdots
}}\right)
\xrightarrow{D_{fzS2}'} 
\left(
\vcenter{
\xymatrix{
T_2& 
T_3&
T_4~~ \cdots
}}\right).
\end{align*}
Of course, $( \mathcal{A} , {\mathbf t})$ is a cone,
where ${\mathbf t} = \{ t_k \in Q ~|~ k >1 \}$.

For the ${\mathcal I}_{fzS2}'$ and $D_{fzS2}'$,
the following proposition is derived.
\begin{proposition} \label{prop5_8}
If there exists the weak classical limit for 
the ${\mathcal I}_{fzS2}'$ and $D_{fzS2}'$,
then the weak classical limit is isomorphic to
$( \mathcal{A} , {\mathbf t})$.
\end{proposition}

\begin{proof}
Let $(L, {\mathbf q}_{lim}) $ be the weak classical limit.
Then there exists a unique Poisson morphism 
$\phi_{\mathcal{A}L} : \mathcal{A} \to L$
such that $t_k = q_{Lk}\circ \phi_{\mathcal{A}L}$ 
for any $q_{Lk}: L \to T_k \in {\mathbf q}_{lim}$ and any $t_k \in {\mathbf t}$.
Recall that this  $\phi_{\mathcal{A}L}$ is a surjective Poisson morphism.
Let us show that this $\phi_{\mathcal{A}L}$ is injective by contradiction.
Suppose that  $\phi_{\mathcal{A}L}$ is not injective.
There are $f,g \in \mathcal{A}$ such that $f\neq g $ and 
$\phi_{\mathcal{A}L}(f)= \phi_{\mathcal{A}L}(g)$.
For any $k$, $t_k( f ) = t_k (g)$ since 
$t_k = q_{Lk}\circ \phi_{\mathcal{A}L}$. 
For sufficiently large $N$, this is contradiction,
and $\phi_{\mathcal{A}L}$ is injective
\end{proof}

As a matter of fact, there is stronger uniqueness 
of the Poisson morphism in the case of fuzzy spheres.

\begin{proposition}\label{fuzzy_unique_mor}
Let ${\mathbf t}$ and  ${\mathbf q}$ be sets of the quantizations 
${\mathbf t} = \{ t_k : \mathcal{A} \to T_k \in Q ~|~ k >1 \}$ and
${\mathbf q} = \{ q_k : B \to T_k \in Q ~|~ k >1 \}$,
where $B$ is arbitrary object in $QW$.
If there is $\phi : B \to \mathcal{A} \in Mor (QW)$ satisfies the following 
commutative diagram:
\begin{align*} 
\xymatrix{
\mathcal{A} \ar[d] \ar[rd] \ar[rrd] & {}
\ar[l]_{\phi}
B
\ar[dl] \ar[d] \ar[rd]^{{}~~~\cdots}
&\\
T_2& 
T_3&
T_4~~ \cdots
} ,
\end{align*}
then $\phi$ is unique.
\end{proposition}

\begin{proof}
 $B$ is a Poisson algebra, since
$ \mathcal{A}$ is a Poisson algebra, and only Poisson algebras
have possibility to be sources of morphisms in $QW$
which target is $\mathcal{A}$. 
Suppose that there exist two morphisms $\phi, \phi' : B \rightarrow \mathcal{A}$
such that $t_k \circ \phi = t_k \circ \phi' = q_k$.
\begin{align*} 
\xymatrix{
\mathcal{A} \ar[d] \ar[rd] \ar[rrd] & {}
\ar[l]_{\phi , \phi'}
B
\ar[dl] \ar[d] \ar[rd]^{{}~~~\cdots}
&\\
T_2& 
T_3&
T_4~~ \cdots
}
\end{align*}
If $\phi \neq \phi'$, there exist $b \in B$ such that 
$\phi (b) \neq \phi'(b)$ and $t_n (\phi (b)) = t_n ( \phi'(b))$ 
for $\forall n \in {\mathbb Z}_{> 0}$.
We denote the coimage of $t_N$ by 
$H_N := \{ f_0+f_{a_1} x^{a_1}+\cdots +f_{a_1\cdots a_{N-1}} x^{a_1}\cdots x^{a_{N-1}}\}$.
For sufficiently large $N$, $\phi (b) , \phi'(b) \in H_N$,
and $t_N (\phi (b)) \neq t_N ( \phi'(b))$.
This is a contradiction, so
the uniqueness of the map $\phi$ is shown.
\end{proof}

\begin{proposition}
For the above ${\mathcal I}_{fzS2}'$ and $D_{fzS2}'$,
the weak classical limit does not exist.
\end{proposition}
\begin{proof}
Suppose that there exists the weak classical limit.
From Proposition \ref{prop5_8}, 
the weak classical limit is isomorphic to
$( \mathcal{A} , {\mathbf t})$.
Let us introduce quantization maps 
$\tilde{t_k}: \mathcal{A} \to T_k ~ (k=2,3,\cdots )$ by
$
\tilde{t}_k(f):=f_0{\bf 1}_k+f_{a_1}X^{a_1}_{k},~
(k=2,3,\cdots )
$
for 
$
f=f_0+f_ax^a+\frac{1}{2}f_{ab}x^ax^b+\cdots \in \mathcal{A}
$.
$( \mathcal{A} , \tilde{\mathbf t})$ is a cone too, where
$\tilde{\mathbf t} = \{ \tilde{t}_k \in Q ~|~ k>1 \}$.
Then there exist the 
unique surjective $\phi : ( \mathcal{A} , \tilde{\mathbf t})
\to ( \mathcal{A} , {\mathbf t})$ such that $t_k \circ \phi = \tilde{t}_k$.
Consider $k \ge 3$ case.  Because
$ t_k \circ \phi (f) = \tilde{t}_k (f)= f_0{\bf 1}_k+f_{a_1}X^{a_1}_{k},$
$\displaystyle \phi (f)= f_0+f_ax^a+ g$,
where $g$ is a polynomial of degree $k$ or higher given as
$\sum_{j\ge k}\frac{1}{j!}g_{a_1\cdots a_j} x^{a_1}\cdots x^{a_j}$, for any $f$. For example, $x^1x^2+ x^2x^1 \notin \mbox{Im} \phi$, and
this contradicts that the $\phi$ is surjective.
\end{proof}



If ${\mathcal I}$ and $D$ are chosen 
to satisfy Theorem \ref{thm_5_1_limit}, the naive classical 
limit is obviously determined, as shown in the examples above. 
On the other hand, to get the weak classical limits 
we have to choose 
appropriate pairs of an index category and a diagram.
If we can not find an appropriate pair of them, 
the weak classical limit is not determined
in general.
Moreover, since any Poisson algebra is a 
candidate for this weak classical limit, 
proving the existence of the weak classical limit 
may be difficult in general.

\subsection{Strong Classical Limit through $(A \downarrow I_{QP})$ }
Apart from the weak classical limit that 
we have examined in the previous subsections, 
it is also natural to assume that 
the classical limit should be determined from the all quantization spaces
corresponding to a single Poisson algebra. 
In particular, it is expected that the fixed Poisson algebra 
is a candidate for the classical limit.
In this subsection, we investigate such classical limits.\\

To prepare for this, we first introduce a way to construct
an index category from a category
with total order.
\begin{definition}
Let $\mathcal{C}$ be a category with total order denoted by 
$c_i \le c_j$ 
$(c_i , c_j \in ob (\mathcal{C}) )$.
$J^\bullet $ is a functor like correspondent 
(but not functor) between 
$\mathcal{C}$ and an index category $\mathcal{I}:=J^\bullet ( \mathcal{C} ) $
defined by
the following conditions:
\begin{enumerate}
\item Between objects of the two categories,
\begin{align*}
J^\bullet :& ob(\mathcal{C})  \xrightarrow{\simeq} ob( \mathcal{I})=  ob(J^\bullet (\mathcal{C}))\\
& {}\ \ \ \ c \ \ \longmapsto i_c := J^\bullet (c)
\end{align*}
is a one to one correspondence.

\item $\forall t_{ij} \in Mor_{\mathcal{C}} (c_i , c_j )$ if
$c_i \le c_j$ then
there exists a $f_{ij}= J^\bullet (t_{ij})  \in Mor_{\mathcal{I}} (i_{c_i} , i_{c_j})$.
$J^\bullet (Id_{c_i})=Id_{i_{c_i}}.$
\item $\forall f_{kl} \in Mor_{\mathcal{I}} (i_{c_k} , i_{c_l})$,
there exists a unique $t_{kl} \in Mor_{\mathcal{C}} ({c_k} , {c_l})$
such that $c_k \le c_l$ and $f_{kl}= J^\bullet (t_{kl})$.
\item  For $c_i \le c_j \le c_k$, if 
$t_{ij} \in Mor_{\mathcal{C}} (c_i , c_j )$,
$t_{jk} \in Mor_{\mathcal{C}} (c_j , c_k )$, and
$t_{ik} \in Mor_{\mathcal{C}} (c_i , c_k )$ satisfy 
$ t_{jk}\circ t_{ij}=t_{ik} $, then 
$J^\bullet ( t_{ik} )=J^\bullet ( t_{jk} )\circ J^\bullet ( t_{ij} ).$
\end{enumerate}

\end{definition}

Note that if $c_k >  c_l$ then 
$Mor_{\mathcal{I}} (i_{c_k} , i_{c_l})= \emptyset$
even when $Mor_{\mathcal{C}} ({c_k} , {c_l}) \neq \emptyset$.
In short, $J^\bullet$ acts to remove morphisms from bigger object to smaller object in $\mathcal{C}$.
This is the reason that $J^\bullet $ is not a functor.

Therefore, we have to show the following.
\begin{proposition}
$\mathcal{I}:=J^\bullet (\mathcal{C})$ is a category.
\end{proposition}

\begin{proof}
The only condition that needs to be checked is consistency regarding
$Mor_{\mathcal{I}} (i_{c_k} , i_{c_l})= \emptyset$ when
$Mor_{\mathcal{C}} ({c_k} , {c_l}) \neq \emptyset$ with $c_k >  c_l$.
For any two $f : X \rightarrow Y$ and $g :  Y \rightarrow Z$ 
their composition $g\circ f : X \rightarrow Z$ have to be a morphism
of $Mor(X,Z)$.
In our situation, if the morphism 
corresponding to $g\circ f$ is removed
without removing the morphism corresponding to $g$ or the morphism corresponding to $f$, then it is impossible for $\mathcal{I}:=J^\bullet (\mathcal{C})$ 
to be a category.
Let $c_X, c_Y , c_Z $ be objects of ${\mathcal C}$ and 
$i_X = J^\bullet ( c_X ),i_Y = J^\bullet ( c_Y ),i_Z = J^\bullet ( c_Z )$
be objects of $\mathcal{I}$. 
For the case $c_X > c_Z$,
$Mor_{\mathcal{I}}(i_X , i_Z) $ becomes $\emptyset$, so if
$Mor_{\mathcal{I}}(i_X , i_Y) \neq \emptyset$ and 
$Mor_{\mathcal{I}}(i_Y , i_Z) \neq \emptyset$ , then it is inconsistent
for $\mathcal{I}$ to be a category.
However, such a situation cannot arise.
The reasons are as follows.
Suppose that $Mor_{\mathcal{I}}(i_X , i_Z) $ becomes $\emptyset$ 
by the condition $c_X> c_Z$.
For the case $c_Y \ge c_X$, $c_Y > c_Z$ also holds, then $Mor_{\mathcal{I}}(i_Y , i_Z) = \emptyset$.
For the case $c_X > c_Y $, $Mor_{\mathcal{I}}(i_X , i_Y) = \emptyset$.
Thus, it is shown that the consistency for compositions of morphisms is guaranteed.
\end{proof}


We introduce a total order to $(A \downarrow I_{QP})$, as follows.
\begin{definition}[total order of $(A \downarrow I_{QP})$]
For $(q_i , T_i) , (q_j , T_j ) \in ob (A \downarrow I_{QP})$ , we define $(q_i , T_i) \le (q_j , T_j )$ by $|\hbar(q_i)|^2 \le |\hbar (q_j)|^2$.
\end{definition}

By this total order, 
an index category
$\mathcal{I}_{\hbar}(A):=  J^\bullet (A \downarrow I_{QP})$ is determined from
$(A \downarrow I_{QP})$.

\begin{align*}
\left(
\vcenter{
\xymatrix{
A \ar[d]_{q_1} \ar[rd]_{q_{2}} \ar[rrd]^{q_{3}\cdots}&{}& {}&\\
T_1 \ar[r]_{t_{12}}& 
T_2 \ar[r]^{t_{23}} &
 \ar@<0.5ex>[l]^{t_{32}} T_3 \ar[r] & \ar@<0.5ex>[l] \cdots
}}
\right)
\rightarrow
\left(
\xymatrix{
i_1 \ar[r]^{f_{12}}& 
i_2 \ar[r]^{f_{23}} &
i_3 \ar[r] &  \cdots
}
\right)= \mathcal{I}_{\hbar}(A)
.
\end{align*}
Here $J^\bullet (q_k , T_k ) = i_k$ and $J^\bullet (t_{kl}) = f_{kl}$
when $|\hbar(q_k)|^2 \le |\hbar(q_l)|^2$.
Next, the diagram 
$D_{\hbar} : \mathcal{I}_{\hbar}(A) \rightarrow (A \downarrow I_{QP})$ 
is defined by 
$D_{\hbar}( i_k ) = (q_k , T_k )$ and $D_{\hbar}( f_{kl} ) = t_{kl}$.
\begin{align*}
\left(
\vcenter{
\xymatrix{
i_1 \ar[r]^{f_{12}}& 
i_2 \ar[r]^{f_{23}} &
i_3 \ar[r] &  \cdots
}}
\right)
\xrightarrow{D_{\hbar} }
\left(
\vcenter{
\xymatrix{
(q_1,T_1) \ar[r]^{t_{12}}& 
(q_2,T_2) \ar[r]^{t_{23}} &
(q_3,T_3) \ar[r] &  \cdots
}}
\right).
\end{align*}
Notice that it is possible to put a natural initial object $( Id_A , A) $
into this $D_{\hbar} (\mathcal{I}_{\hbar}(A) ) $.
\begin{definition}
$A/QW_{\hbar}$ is a category defined by
\begin{align*}
ob( A/QW_{\hbar} ):= ob ( D_{\hbar} (\mathcal{I}_{\hbar}(A) )) \cup \{( Id_A , A)\}\end{align*}
and
\begin{align*}
Mor_{A/QW_{\hbar}}((q_i , T_i) , (q_j , T_j)) =&
Mor_{D_{\hbar} (\mathcal{I}_{\hbar}(A) )}((q_i , T_i) , (q_j , T_j)) \\
& \mbox{for } ~ \forall (q_i , T_i) , (q_j , T_j)\in ob ( D_{\hbar} (\mathcal{I}_{\hbar}(A) )) , \\
Mor_{A/QW_{\hbar}}((Id_A , A) , (q_i , T_i))=& \{ q_i \} ~~ \mbox{for}~~ 
\forall
(q_i , T_i) \in ob ( D_{\hbar} (\mathcal{I}_{\hbar}(A) )).
\end{align*}
\end{definition}

\begin{align*}
\vcenter{
\xymatrix{
(Id_A ,A ) \ar[d]_{q_1} \ar[rd]_{q_{2}} \ar[rrd]^{q_{3}\cdots}&{}& {}&\\
(q_1,T_1) \ar[r]_{t_{12}}& 
(q_2 ,T_2) \ar[r]_{t_{23}} &
 (q_3 , T_3) \ar[r] &  \cdots
}}
.
\end{align*}

Note that the index $\mathcal{I}_{\hbar}(A)$ and the diagram $ D_{\hbar}$
are still valid in $A/QW_{\hbar}$.

\begin{proposition}\label{prop_stronglimit}
$A/QW_{\hbar}$ has a limit $(( Id_A , A), {\mathbf q})$ 
for the index category
$\mathcal{I}_{\hbar}(A)$ with the diagram $ D_{\hbar}$.
Here, $(( Id_A , A), {\mathbf q})$
is a cone of 
$D_{\hbar} (\mathcal{I}_{\hbar}(A) )$
with ${\mathbf q} = \{q_i \in Q ~|~ s(q_i)= A \}$.
\end{proposition}

\begin{proof}
$( Id_A , A)$ is the initial object in $A/QW_{\hbar}$.
By definition, $(( Id_A , A), {\mathbf q})$
is a cone of 
$D_{\hbar} (\mathcal{I}_{\hbar}(A) )$.
Because every morphism $t_{ij} : (q_i , T_i) \rightarrow (q_j , T_j)$
satisfies commutativity $q_j = t_{ij}\circ q_i$ 
by the definition of $(A \downarrow I_{QP})$.
In $A/QW_{\hbar}$, there is at most one morphism between any two objects.
So, $(( Id_A , A), {\mathbf q})$ is one of the candidates 
of the limit.
Next, we prove that there are no other candidates for the limit other 
than $(( Id_A , A), {\mathbf q})$, by contradiction.
If the other candidate of the limit of the diagram $ D_{\hbar}$ exists,
it is given by a cone with the form $(( q_k , T_k) , {\mathbf t}_{k -} )$,
where $q_k : A \rightarrow T_k $ is in $Q$ and 
${\mathbf t}_{k -} = \{ t_{kl}: q_k \rightarrow q_l= t_{kl}\circ q_k
~|~ t_{kl}\in Mor_{QP}(T_k , T_l) \}$.
From Proposition \ref{prop7}, there exist $q_x : A \to T_x \in Q$
such that $|q_x|^2 < |q_k|^2$. $(q_x , T_x )$ is an object in $D_{\hbar} (\mathcal{I}_{\hbar}(A) )$, and $( q_k , T_k)> (q_x , T_x )$.
Then $Mor_{A/QW_{\hbar}}(( q_k , T_k),(q_x , T_x ))= \emptyset$, and 
this contradicts $(( q_k , T_k) , {\mathbf t}_{k -} )$ is a cone.
Therefore, we find that 
$(( Id_A , A), {\mathbf q})$ is the limit.
\end{proof}

From this proposition, we found that for any Poisson algebra $A$ in $QW$
the limit for the index category
$\mathcal{I}_{\hbar}(A)$ with the diagram $ D_{\hbar}$
is determined without additional information.
We call this limit $(( Id_A , A), {\mathbf q})$ 
``strong classical limit of $A$".

This Proposition \ref{prop_stronglimit}
shows that the strong classical limit satisfies 
at least the properties that the limit should have that we wanted 
in this subsection.

\subsection{Example of Strong Classical Limit}

Let us consider $\mathcal{A}$ we studied in the fuzzy sphere
in Appendix \ref{appendix_fuzzy}.
The index $\mathcal{I}_{\hbar}(\mathcal{A})$, the diagram 
$ D_{\hbar}(\mathcal{I}_{\hbar}(\mathcal{A}))$ 
and $\mathcal{A}/QW_{\hbar}$ are 
determined 
without extra information.
We obtain the following fact.
\begin{example}\label{ex_calA_stlim}
Let $\mathcal{A}$ be a Poisson algebra studied
in Appendix \ref{appendix_fuzzy}.
The strong classical limit of $\mathcal{A}$ is 
$((Id_\mathcal{A}, \mathcal{A}) , {\mathbf q} )$, where
${\mathbf q} =\{ q_i \in Q ~| ~ s(q_i) = \mathcal{A} \} $.
\end{example}
From Proposition \ref{prop_stronglimit}
there is no need to prove 
this Example \ref{ex_calA_stlim}.
However, this claim can be shown just 
from the properties of some quantization maps 
of matrix regularization. 
This is so suggestive that we will also provide the proof below.

\begin{proof}
As we saw in the proof of Proposition \ref{prop_stronglimit},
$((Id_\mathcal{A}, \mathcal{A}) , {\mathbf q} )$ is a cone 
for the diagram $ D_{\hbar}$ in $\mathcal{A}/QW_{\hbar}$.
We show that there are no other cones.

Suppose that there is another cone.
Then, the cone is given as some 
$((q_{cone} , T_{cone} ), {\mathbf t}_{cone,-})$ where 
$q_{cone} \in Q $ with $s(q_{cone})= \mathcal{A}$ and
${\mathbf t}_{cone,-} = \{ t_{cone, k } : (q_{cone} , T_{cone} ) \rightarrow 
(q_{k} , T_{k} ) ~|~ q_k \in Q ,~s(q_k)= \mathcal{A},
 q_k = t_{cone, k }\circ q_{cone} \} $.
Consider $t_{cone , 2} \in Mor ( (q_{cone} , T_{cone} ) , (t_{2} , T_{2}))$,
where $t_2 \in Q$ is defined in 
 Appendix \ref{appendix_fuzzy}.
Since $t_{cone , 2}$ is an injection and 
$T_{2}$ is a $4$-dimensional vector space, $\dim T_{cone} \le 4$ is derived. 
For $t_3 : \mathcal{A} \rightarrow T_3$ defined in 
 Appendix \ref{appendix_fuzzy},
$\dim \mbox{\rm Im } t_3 > 4$.
On the other hand, $\dim \mbox{\rm Im } t_{cone, 3} \le 4$ for 
$t_{cone, 3} : T_{cone}\rightarrow T_3 $.
Therefore $t_3 \neq t_{cone,3}\circ q_{cone} $.
This is contradiction to the assumption that 
$((q_{cone} , T_{cone} ), {\mathbf t}_{cone,-})$ is a cone.

Thus, it is shown that $((Id_\mathcal{A}, \mathcal{A}) , {\mathbf q} )$
is the only cone and the limit 
for the diagram $ D_{\hbar}$ in $\mathcal{A}/QW_{\hbar}$.
\end{proof}

\section{One Attempt to Address the Inverse Problem}\label{sect5}

As mentioned at the beginning of this article,
the inverse problems of quantizations are important for physics related to 
noncommutative geometries.
The inverse problem of quantizations is the problem of how to determine the classical limits (classical manifolds) based on the information 
of the spaces appearing in the quantizations.
In general, classical limits are not uniquely determined 
from algebras as target spaces of quantization maps 
in a naive way, so we need to experiment 
with various subtle techniques.
This problem is closely related to the question of how we can identify
manifolds in the membrane theory, thus, 
it was often investigated in that context.
For example, in \cite{berezin1,bordemannA,Chu:2001xi}
it is discussed how to construct the classical limit in geometric quantization.
In \cite{shimada,berenstein,Schneiderbauer}, how membrane topology is distinguished in the context of
matrix regularization.
In \cite{ishiki1,ishiki2,asakawa}, 
inverse problems of Berezin-Toeplitz quantizations 
are discussed.\\

Unlike in the past, the purpose of this section is to study 
this inverse problem as a problem of the classical limit 
discussed in the previous section.
Until now, we have studied whole world of the quantizations 
of all Poisson algebras. 
In particular, the naive classical limit 
and the weak classical limits introduced in 
Subsection \ref{naive_classical_limit} and 
Subsection \ref{subsec_weak_limit} are
 chosen from among all Poisson algebras in $QW$.
To determine the naive classical limit or weak classical limit from 
the spaces appearing in the quantizations,
we are going to examine how to construct a sequence of objects in $QW$
and are going to discuss how to give the sequence 
by a process found in physics.

\subsection{Matrix Regularization and Inverse Problem}

The goal of this section is making an example 
of a method to obtain a Poisson algebra as a classical limit
from a noncommutative associative algebra.
Here we make this example using a fuzzy sphere as a role model.
In this section, the commutative ring $R$ is fixed to ${\mathbb C}$, again. 
\\

At first, let $\mathfrak{g}$ be a Lie algebra 
as a subset of a noncommutative algebra over ${\mathbb C}$ 
whose commutator is defined by its associative product
i.e. $[ a, b ]:= ab - ba$.
Only semisimple Lie algebras are treated in this section.
This $\mathfrak{g}$ is the origin to get the sequence of objects in $QW$
that may provide a classical limit.
Such sequence of objects in $QW$ is named 
the quantization family for the weak classical limit 
in Subsection \ref{subsec_weak_limit}.
Let $e=\{  e_1 , e_2 , \cdots ,e_d \}$ be a fixed base of $\mathfrak{g}$
satisfying commutation relations $[ e_i , e_j ]= f_{ij}^k e_k $,
where $f_{ij}^k$ are structure constants of $\mathfrak{g}$. 
For this Lie algebra $\mathfrak{g}$ we introduce a sequence of 
representation $\rho^{\mu} : \mathfrak{g} \rightarrow gl(V_\mu ) $ and
a sequence of numbers $\hbar(\mu )\neq 0$, ($\mu = 1,2,3,\cdots $).
Here $V_\mu$ is a finite dimensional vector space chosen as appropriate.
We denote the corresponding basis of $e$ by
\begin{align}\label{basis_matrix}
e^{(\mu)} =\{ 
\hbar(\mu) \rho^{\mu} (e_1 ), \hbar(\mu) \rho^{\mu} (e_2 ), 
\cdots , \hbar(\mu) \rho^{\mu} (e_d ) \}
= \{ 
e_1^{(\mu)} , e_2^{(\mu)} , \cdots  ,e_d^{(\mu)} \}.
\end{align}
Then they satisfy
\begin{align}
[ e_i^{(\mu)} , e_j^{(\mu)} ]= \hbar( \mu )f_{ij}^k e_k^{(\mu)} .
\end{align}
The Lie algebra $\rho^{\mu}(\mathfrak{g})$ 
or $\langle e^{(\mu)} \rangle$ are constructed by this basis.


For each $\mu$, we have relations that are originated in Casimir invariants.
To characterize a classical limit, we can use these relations.
For the later convinience, we denote them by
\begin{align}\label{ideal_relation}
 f_1^{\mu}(e^{(\mu)}, \nu^{(\mu)}, \hbar( \mu )),
\cdots ,
f_{N_\mu}^{\mu}(e^{(\mu)}, \nu^{(\mu)}, \hbar( \mu )) ,
\end{align}
where $\nu^{(\mu)}= \{ \nu_1^{(\mu)} , \nu_2^{(\mu)} , \cdots \}$
is a set of parameters, and $N_\mu \le rank \mathfrak{g}$.
We will use 
these relations to induce corresponding relations 
in the Poisson algebra, later.
We will be back to this subject, at the end of this subsection.

Next, let
$T_{\mu}$ be the vector space that is $\langle e^{(\mu)} \rangle$
forgetting multiplication structure.
We choose a basis of $T_{\mu}:=\langle e^{(\mu)} \rangle$,
$E_1, E_2, \dots, E_D$, as polynomials of $e^{(\mu)}$.
The highest degree is denoted by $n_\mu$, i.e.
$n_\mu = \max \{ {\rm deg}E_1, \cdots , {\rm deg}E_D \}$.
For later use, let us consider symmetrized polynomial 
$$c_{i_1, \cdots , i_k} e^{(\mu)}_{i_1} \cdots e^{(\mu)}_{ i_k} 
= c_{i_1, \cdots , i_k}e^{(\mu)}_{(i_1, \cdots ,i_k)}
\qquad (k= 0,1,2,\cdots ,n_\mu ),
$$ where the fixed coefficients $c_{i_1, \cdots , i_k} \in {\mathbb C}$
are complete symmetric, and 
\begin{align*}
e^{(\mu)}_{(i_1, \cdots ,i_k)}:=
\frac{1}{k!}\sum_{\sigma \in Sym(k)}
e^{(\mu)}_{i_{\sigma(1)}} \cdots e^{(\mu)}_{ i_{\sigma(k)}} .
\end{align*}
Here $Sym(k)$ is a symmetric group.
Obviously, each $e^{(\mu)}_{(i_1, \cdots ,i_k)} $
is given as a linear combination of $E_1, E_2, \dots, E_D$
and its degree is smaller than or equal to $n_\mu$.
For $m > n_\mu$ $e^{(\mu)}_{(i_1, \cdots ,i_m)}$
is written by some linear combination of lower degree polynomial,
then $\hbar(\mu)$ appear 
because of the definition (\ref{basis_matrix}).
From the set of representation  $\rho^{\mu}$ (and relations $f_i^{\mu}$),
we obtain sets of $T_{\mu}$.
We denote a set of all $T_{\mu}$ by $\{ T_{\mu} \}$.
To show that $\{ T_{\mu} \}$ is a sequence of
objects in $QW$, we have to find one Poisson algebra and to make 
quantization maps from the Poisson algebra to 
$\{ T_{\mu} \}$.\\

The next step we introduce a way to obtain one Poisson algebra
from a Lie algebra $\mathfrak{g}$ as a candidate of a weak classical limit.
There is a well-known way known as Kirillov-Kostant Poisson bracket,
that is the way constructing a Poisson algebra.
(See also \cite{Weinstein_Lu,SemenovTianShansky,Alekseev:1993qs}.)
We focus the following fact (See \cite{matrix1,Kostant}), here.
\begin{theorem}\label{LiePoissonThm}
Let $\mathfrak{g}$ be a $d$-dimensional Lie algebra.
Let $e=\{ e_1 , e_2 , \cdots , e_d \}$ be a basis of $\mathfrak{g}$
satisfying commutation relations $[ e_i , e_j ]= f_{ij}^k e_k $. 
Let $x=(x^1, x^2, \cdots , x^d)$ be commutative variables.
We obtain a Poisson algebra $(\mathbb{C}[x]  , \cdot , \{ ~ , ~ \})$ by
\begin{align}
\{ f , ~ g\}:= f \omega g := f \overleftarrow{\partial}_i \omega_{ij} 
\overrightarrow{\partial}_j g := ({\partial}_i f) \omega_{ij} 
({\partial}_j g) ,
\end{align}
where $\displaystyle \partial_i = \frac{\partial}{\partial x^i}$ and
$\omega_{ij} = f_{ij}^k x^k $.
\end{theorem}
The proof is given in \cite{matrix1} but for readers convenience 
we give the proof here.
\begin{proof}
$\{ f , ~ g\}= - \{ g , ~ f\}$ is followed from $f_{ij}^k = - f_{ji}^k$.
The Leibniz's rule and the bilinearity are 
trivially satisfied by the definition.
By direct calculations, the Jacobi identity is obtained as follows.
\begin{align*}
&\{ \{f_1 , f_2 \} , f_3 \} + \{ \{f_2 , f_3 \} , f_1 \}+
\{ \{f_3 , f_1 \} , f_2 \}\\
&= (\partial_i f_1 ) (\partial_j f_2 )(\partial_l f_3 )
\{
(\partial_k \omega_{ij}) \omega_{kl} +
(\partial_k \omega_{jl}) \omega_{ki} +
(\partial_k \omega_{li}) \omega_{kj} 
\}.
\end{align*}
By $\omega_{ij} = f_{ij}^k x^k $,
\begin{align*}
(\partial_k \omega_{ij}) \omega_{kl} +
(\partial_k \omega_{jl}) \omega_{ki} +
(\partial_k \omega_{li}) \omega_{kj}
= x^m ( f_{ij}^k f_{kl}^m + f_{jl}^k f_{ki}^m + f_{li}^k f_{kj}^m )=0
\end{align*}
The last equality follows from the Jacobi identity of the Lie bracket
$[ e_i , e_j ]= f_{ij}^k e_k $.
\end{proof}
We denote this $(\mathbb{C}[x]  , \cdot , \{ ~ , ~ \})$ by
$A_\mathfrak{g}$.
From the above discussions, we found that we can obtain 
a sequence of Lie algebras and a Poisson algebra from a
fixed single Lie algebra.\\

Next, let us construct quantization maps 
from the Poisson algebra $A_\mathfrak{g}$ 
we have just obtained in Theorem \ref{LiePoissonThm}
to the vector spaces in $\{T_\mu \}$ that generate original algebras.
We define linear function
$q_\mu : A_\mathfrak{g} \rightarrow T_\mu $
by 
\begin{align} \label{q_mu}
\sum_k f_{i_1, \cdots , i_k} x^{i_1} \cdots x^{ i_k} \mapsto 
\sum_k^{n_\mu} f_{i_1, \cdots , i_k} e^{(\mu)}_{i_1} \cdots e^{(\mu)}_{ i_k}
= \sum_k^{n_\mu} f_{i_1, \cdots , i_k} e^{(\mu)}_{(i_1, \cdots , i_k)} ,
\end{align}
where $f_{i_1, \cdots , i_k} \in \mathbb{C}$ is completely symmetric,
and we assume that the multiplicative identity of  $A_\mathfrak{g}$ 
maps to the unit matrix in $T_\mu$.
Then this correspondence $q_\mu$ satisfies the following.
\begin{theorem}\label{sect5_thm_2}
Let $q_\mu : A_\mathfrak{g} \rightarrow  T_\mu = \langle e^{(\mu)}\rangle$ be a linear function defined as (\ref{q_mu}). Then it satisfies
\begin{align*}
[q_{\mu} ( f ) , q_{\mu} ( g )]
= \hbar ({\mu}) q_{\mu} ( \{ f , g \} ) + \tilde{O}(\hbar({\mu}))
\end{align*}
for $f,g \in A_\mathfrak{g} $. In other words, $q_\mu \in Q$.
\end{theorem}

\begin{proof}
We put $f= \sum_k f_{i_1, \cdots , i_k} x^{i_1} \cdots x^{ i_k}$
and $g= \sum_l g_{j_1, \cdots , j_l} x^{j_1} \cdots x^{ j_l}$.
\begin{align*}
[q_{\mu} ( f ) , q_{\mu} ( g )]
&= 
\sum_{k,l}^{n_\mu} f_{i_1, \cdots , i_k} g_{j_1, \cdots , j_l}
[ e^{(\mu)}_{i_1} \cdots e^{(\mu)}_{ i_k}, 
 e^{(\mu)}_{j_1} \cdots e^{(\mu)}_{ j_l}] \\
&= 
\sum_{k,l=1}^{n_\mu} f_{i_1, \cdots , i_k} g_{j_1, \cdots , j_l}
\sum_{n,m} 
 e^{(\mu)}_{i_1} \cdots [ e^{(\mu)}_{ i_n}, 
 e^{(\mu)}_{j_m}] \cdots e^{(\mu)}_{ j_l} .
\end{align*}
Here using
$
[ e^{(\mu)}_{ i_n}, 
 e^{(\mu)}_{j_m}] = \hbar( \mu )f_{i_n j_m}^k e_k^{(\mu)}
$,
\begin{align}
[q_{\mu} ( f ) , q_{\mu} ( g )]
&= \hbar(\mu )
\sum_{k,l=1}^{n_\mu} f_{i_1, \cdots , i_k} g_{j_1, \cdots , j_l}
\sum_{n,m} f_{i_n j_m}^p 
 e^{(\mu)}_{i_1} \cdots e^{(\mu)}_{p} \cdots e^{(\mu)}_{ j_l} 
\notag \\
&= \hbar(\mu ) \sum_{k,l=1}^{n_\mu} f_{i_1, \cdots , i_k} g_{j_1, \cdots , j_l}
\sum_{\substack{1\le n \le k\\
1\le m \le l}} f_{i_n j_m}^p 
e^{(\mu)}_{(i_1, \cdots , \hat{i}_n ,\cdots,\hat{j}_m, \cdots,j_l, p) }
+\tilde{O}(\hbar^{1+\epsilon}(\mu))
. \label{thm6_2_1}
\end{align}
Here the index $\{ i_1, \cdots , \hat{i}_n ,\cdots,\hat{j}_m, \cdots,j_l ,p \}$
means $\{  i_1, \cdots , i_k , j_1, \cdots , j_l , p\} - \{ i_n , j_m \}$,
so $e^{(\mu)}_{(i_1, \cdots , \hat{i}_n ,\cdots,\hat{j}_m, \cdots,j_l,p )}$
is a degree $k+l -1$ polynomial.
$ \tilde{O}(\hbar^{1+\epsilon}(\mu))$ appeared when 
$e^{(\mu)}_{i_1} \cdots e^{(\mu)}_{p} \cdots e^{(\mu)}_{ j_l} $
were sorted in the symmetric order.
On the other hand,
\begin{align*}
\{ f , g \}&= 
\sum_{k,l=1} 
f_{i_1, \cdots , i_k} g_{j_1, \cdots , j_l}
\sum_{\substack{1\le n \le k\\
1\le m \le l}} f_{i_n j_m}^p x^{i_1} \cdots x^{p} \cdots x^{ j_l} .
\end{align*}
Note that the degree of $x^{i_1} \cdots x^{p} \cdots x^{ j_l}$ is
$k+l-1$.
Then 
\begin{align}\label{thm6_2_2}
q_\mu (\{ f , g \} )=\sum_{\substack{1\le k,l \\ k+l \le n_\mu +1}}
f_{i_1, \cdots , i_k} g_{j_1, \cdots , j_l}
\sum_{\substack{1\le n \le k\\ 1\le m \le l}} f_{i_n j_m}^p 
e^{(\mu)}_{(i_1, \cdots , \hat{i}_n ,\cdots,\hat{j}_m, \cdots,j_l, p) }.
\end{align}
Let us subtract $\hbar$ times (\ref{thm6_2_2})  from (\ref{thm6_2_1}):
\begin{align*}
&[q_{\mu} ( f ) , q_{\mu} ( g )]- \hbar(\mu ) q_\mu (\{ f , g \} ) \\
&= \hbar(\mu ) \!\!\!\!\!\!\!\!
 \sum_{\substack{2\le k,l \le n_\mu \\ n_\mu +2 \le k+l \le 2 n_\mu}}
\!\!\!\!\!\!\!\!
f_{i_1, \cdots , i_k} g_{j_1, \cdots , j_l}
\sum_{\substack{1\le n \le k\\
1\le m \le l}} f_{i_n j_m}^p 
e^{(\mu)}_{(i_1, \cdots , \hat{i}_n ,\cdots,\hat{j}_m, \cdots,j_l, p) }+
\tilde{O}(\hbar^{1+\epsilon}(\mu)).
\end{align*}
%
Note that 
the degree of $e^{(\mu)}_{(i_1, \cdots , \hat{i}_n ,\cdots,\hat{j}_m, \cdots,j_l, p) }$ in the right hand side is bigger than or equal to
$n_\mu +1$.
Recall that $n_\mu$ is determined as the highest degree 
of polynomials of
$E_i (i=1,\cdots ,D)$
that constitute basis of $T_\mu$.
So, any degree $m$ polynomial of $e^{(\mu)}$
with $m > n_\mu$ 
is represented by some linear combination of polynomials
whose degree is smaller than or equal to $n_\mu$.
From the definition (\ref{basis_matrix}),
such linear combinations are $\tilde{O}(\hbar({\mu}))$.
Then we obtain
\begin{align*}
[q_{\mu} ( f ) , q_{\mu} ( g )]
&=
\hbar({\mu}) q_{\mu} ( \{ f , g \} ) + \tilde{O}(\hbar({\mu})).
\end{align*}
\end{proof}
Thus, a sequence of quantization maps $\{ q_\mu \}\subset Q$ is obtained, 
and we find $\{ T_\mu \}$ is a sequence of objects in $QW$.\\
\bigskip

Other candidates for Poisson algebras exist besides 
$A_\mathfrak{g} $.
Using ${I}$ as an ideal of $A_\mathfrak{g} $ generated by
relations that is invariants under acting Poisson brackets,
we can introduce $A_\mathfrak{g} / { I}$.
Even in this case, we expect to be able to construct $\{ T_{\mu} \}$ 
and  $\{ q_{\mu} \}$ in the same way as in the $A_\mathfrak{g} $ case.
In fact, 
${q}_{\mu}$ is realized as $t_k$ 
for a fuzzy sphere case 
in Appendix \ref{appendix_fuzzy}.

\begin{example}\label{ex5_1}
Let $e=\{ e_1 , e_2 , e_3 \}$ be a fixed base of $\mathfrak{su}(2)$
satisfying commutation relations $[ e_i , e_j ]= i\epsilon_{ijk} e_k $.
Consider a sequence of 
spin $j$ representation  
$\rho^{\mu_j} : \mathfrak{g} \rightarrow gl(V_{\mu_j} ) $ (${\rm dim} V_{\mu_j} = 2j +1$).Then $e^{(\mu_j)}$ satisfies
$[ e^{(\mu_j)}_i , e^{(\mu_j)}_j ]= i\hbar({\mu_j}) \epsilon_{ijk} e^{(\mu_j)}_k $,
and a Casimir relation 
\begin{align*}
f^{\mu_j}( e^{(\mu_j)}, R , \hbar({\mu_j}) ) 
= (e^{(\mu_j)}_1 )^2 + (e^{(\mu_j)}_2 )^2 +(e^{(\mu_j)}_3 )^2 - R^2 Id 
\end{align*}
is imposed for each $\rho^{\mu_j}$.
Here $R/ \hbar({\mu_j}) = \sqrt{j(j+1) }$.
From this sequence we obtain the sequence of $T_{\mu}$ in Subsection
\ref{naive_classical_lim_FzS2} or 
in Subsection \ref{ex_weak_lim}. 
The corresponding Poisson algebra of $\mathfrak{su}(2)$ 
determined by Theorem \ref{LiePoissonThm}
is differ from 
${\mathcal A}$ because of the existence of the relation. 
However, quantization maps $t_k$ in Subsection
\ref{naive_classical_lim_FzS2} or 
Subsection \ref{ex_weak_lim}
is constructed in the same manner with 
$q_\mu$ in Theorem \ref{sect5_thm_2}.
Furthermore, as we saw in Subsection
\ref{naive_classical_lim_FzS2}, 
this sequence gives the naive classical limit 
${\mathcal A}$, 
when its diagram includes ${\mathcal A}$.
\end{example}
\bigskip

As above, starting from a Lie algebra 
we obtain a sequence of corresponding objects $\{ T_{\mu} \}$ in $QW$
and some corresponding Poisson algebras.
We already know that the diagram $\{ T_{\mu} \}\cup {A_\mathfrak{g} }$ 
gives at least a naive classical limit.
This indicates that the above procedure 
may provide one new approach to the inverse problem.
\\

Is it possible to derive a quantization family,
a sequence like $\{ T_{\mu} \}$ that gives a naive classical limit
or a weak classical limit,
from the principle of least action in terms of physics?
Here we mention only one attempt to give such a quantization family.
Let us consider a sequence of matrix models.
Using $N\times N$ matrixes $X^N_{\mu} ~(\mu = 1, \cdots , D, ~ N \in {\mathbb N} )$
and a mass $\hbar(N)$, we consider the action
\begin{align}
S_N(\hbar^2(N)) = \mbox{\rm{tr}} \left(
\frac{1}{4} [ X^N_{\mu} , X^N_{\nu}]^2 + \hbar^2(N) \frac{1}{2} X^{N \mu}X^N_{\mu}
\right)+ \sum_i^k \langle \lambda_i^N , f^{(N)}_i(X^N, \nu) \rangle ,
\end{align}
where we take contraction 
by the Killing metric with the structure constants 
$f_{\mu \rho}^\tau$ of $\mathfrak{g}$, i.e. $g_{\mu \nu}= f_{\mu \rho}^\tau f_{\nu \tau}^\rho$, and $\langle - , - \rangle$
represents the inner product with respect to this Killing metric.
$f^{(N)}_i(X^N, \nu )~ (i= 1,\cdots ,k \le rank \mathfrak{g})$ 
are relations given by Casimir invariants like 
$f_i^{\mu}(e^{(\mu)}, \nu^{(\mu)}, \hbar(\mu))$ in (\ref{ideal_relation})
with some parameters $\nu$.
For simplicity, we assume the relations $\{f^{(N)}_i \}$ and parameters $\nu$
do not depend on $N$ i.e. $f^{(N)}_i(x , \nu ) = f^{(J)}_i (x , \nu )$ 
for all $N , J \in {\mathbb N}$ and $x \in {\mathbb C}^D$.
$\sum_i \langle \lambda_i^N , f^{(N)}_i(X^N, \nu) \rangle$ 
is a gauge fixing term
and $\lambda_i^N $ is a Lagrange multiplier.
The equation of motion with each $\lambda_i^N$ is 
$f^{(N)}_i(X^N, \nu)=0$.

Then the equation of motion  obtained by the variation of $X^N_{\nu}$ 
with $\lambda_i^N = 0$ is given by
\begin{align*}
[X^{N \mu} ,  [ X^N_{\mu} , X^N_{\nu}]]= \hbar^2(N) X^N_{\nu} .
\end{align*}
If $X^N_{\mu}$ is a base of a representation of $\mathfrak{g}$ such that
$[ X^N_{\mu} , X^N_{\nu}]=\hbar(N) f_{\mu \nu}^\tau X^N_{\tau}$,
the left hand side of the equation of motion is calculated as
\begin{align*}
[X^{N \mu} ,  [ X^N_{\mu} , X^N_{\nu}]] &=
\hbar(N)f_{\mu \nu}^\rho g^{\mu \tau}[X^N_\tau , X^N_\rho ]
=\hbar^2(N) f_{\mu \nu}^\rho g^{\mu \tau} f_{\tau \rho}^\sigma X^N_\sigma
\\
&= \hbar^2(N){\delta_{\nu}}^\sigma X^N_\sigma = \hbar^2(N)X^N_\nu .
\end{align*}
Therefore, we found that
such set of generators $\{ X^N_{\mu} \}$ which consist a representation of $\mathfrak{g}$ is a 
solution of the equation of motion of the action $S_N$.
(In the previous works by Ishii et al. \cite{Ishii:2008tm} and
by Kim et al. \cite{Kim:2003rza}, actions that derive similar equation of motion
are investigated.)
We denote these generators by $e^{(N)}$.
A vector space $V_N$ spanned by the solutions $e^{(N)}$ is determined. 
A subset of natural number 
$SN \subset {\mathbb N}$ determine a sequence of actions
$\{S_N \}:=\{S_N ~|~ N\in SN \subset {\mathbb N} \}$.
Then we obtain sequence of $V_N$ from this $\{S_N \}$ similarly.
The sequence $\{\langle V_N \rangle \}$ is a sequence of objects $QW$.
\\

The way to get one Poisson algebra $A_{\mathfrak{g}}$
from generators $e^{(N)}$ is
given in Theorem \ref{LiePoissonThm}.
For the case with relations 
$\{f^{(N)}_i(X^N, \nu^{(N)})  \}$,
a Poisson algebra with relations $\{ f^{(N)}_i(x, \nu^{(N)}) \}$
is also expected. (Recall that $f^{(N)}_i$ does not depend on $N$.)
We assume the existence of the Poisson algebra and 
denote the Poisson algebra 
by $A_{\mathfrak{g}}/ { I}$,
where 
$I= \langle f_1^{(N)}(x , \nu),
\cdots ,
f_{k}^{(N)}(x, \nu) \rangle$.
To obtain the classical limit, 
we need a sequence of quantization maps from 
the Poisson algebra to objects.

While a detailed examination of the following discussion is left for future work, an outline of the methods expected to be accomplished is provided here.
Let us find the way to obtain the set of 
these quantization maps.
From Theorem \ref{sect5_thm_2}
quantization $q_N \in Q$
from $A_{\mathfrak{g}}/ I$ to $\langle V_N \rangle $ 
is defined similar to (\ref{q_mu}).
Using $\{ q_N ~|~ N \in SN \}$, we define a new 
quantization sequence 
$\{q_N' \} = \{ q_N'=q_N, q_J'=q_J , q_K'= q_K, \cdots \}\subset Q $
as appropriate such that
$\ker q_N' \supsetneq \ker q_J'$ if $N < J$.
(The similar condition appears in Fuzzy sphere case.)\\

Finally, using the objects $A_{\mathfrak{g}}/ I$ and 
$ q_N' (A_{\mathfrak{g}}/ I)$ and 
the morphisms $\{q_N' \} $ in $QW$, 
as in Section \ref{naive_classical_lim_FzS2},
we can construct the index category and its diagram.
We denote the diagram by $D_{\mathfrak{g}}$.
The constructing process is parallel to 
the one in Section \ref{naive_classical_lim_FzS2}.
Then, the diagram $D_{\mathfrak{g}}$ is obtained by (\ref{6_diagram1})
and the naive classical limit is given by 
$A_{\mathfrak{g}}/ I$.
\begin{align}\label{6_diagram1}
\left(
\vcenter{
\xymatrix{
1 \ar[d]^{\hat{N}} \ar[rd]^{\!\!\!\hat{J}} \ar[rrd]^{\hat{K} ~~\cdots}& {}&\\
N& 
J&
K~~ \cdots
}}\right)
\xrightarrow{D_{\mathfrak{g}}} 
\left(
\vcenter{
\xymatrix{
A_{\mathfrak{g}}/ I \ar[d]^{q_N'} \ar[rd]^{\!\!\! q_J'} \ar[rrd]^{q_K' ~~\cdots}& {}&\\
q_N' (A_{\mathfrak{g}}/ I )& 
q_J' (A_{\mathfrak{g}}/ I )&
q_K' (A_{\mathfrak{g}}/ I )~~ \cdots
}}\right).
\end{align}
\\

Let us summarize the process to obtain the classical limit.
\begin{enumerate}
\item Prepare the sequence of actions $\{S_N \}$.
\item A sequence of the solution $e^{(N)} = \{e^{(N)}_1, \cdots , e^{(N)}_d \}$
and the sequence of vector spaces $\{V_N \}$
, where $e^{(N)} $ is a basis of $V_N$ are obtained.
\item When $A_{\mathfrak{g}}/ I$ is 
given as a Poisson algebra from $\{ e^{(N)} \}$, 
a sequence of quantization $\{q_N' \} \subset Q$, and $ q_N' (A_{\mathfrak{g}}/ I)$ are constructed.
\item Objects $A_{\mathfrak{g}}/ I$, 
$ q_N' (A_{\mathfrak{g}}/ I)$ and  morphisms $q_N'$
induce an index category and its diagram $D_{\mathfrak{g}}$.
\item The naive classical limit is determined as $A_{\mathfrak{g}}/ I$. 
\end{enumerate}

It should be noted that, in contrast to conventional classical mechanics, 
the classical solution gives a basis for the vector space, 
and the set of basis (interpreted as the set of vector spaces) derives
the sets of objects and morphisms in $QW$ 
that may determine the classical limit.

The process discussed here is 
an example of approaching an inverse problem 
using the quantization world $QW$ 
through a physics method.
Since the method is based on a generalization of the fuzzy sphere, 
it works well for the case of fuzzy spheres, 
but it is not known at this time whether 
it will work for other manifolds.
It is another future task to investigate this point.

\section{Summary}\label{sect6}

We constructed a category $QW$ that is composed of the all 
quantizations of all Poisson algebras.
The characteristic of this category is that 
quantization is treated as a linear map not to an algebra, 
but to a module, which is a subset of some algebra.
We defined what is required in this formulation and saw next that the definition works well.
First, equivalence of quantization was defined. And it was shown that iff there is a pair of equivalent quantizations by the definition  
then there exists a pair of isomorphisms of Poisson algebras and the algebras 
generated by the images of the quantization maps.
We also considered the category of quantizations 
for a fixed single Poisson algebra and 
discussed its classification of quantizations.
Next, we defined the classical limits.
Three types of classical limits were introduced: naive classical limits, 
weak classical limits and strong classical limits.
The naive classical limit was defined 
in the context of 
category theory.
It is a limit in category $QW$
whose 
vertex of the cone is in Poisson algebras.
As concrete examples, we introduced the naive classical limits 
for deformation quantizations of polynomials 
with the Moyal products and for matrix regularization of spheres.
The weak classical limit was defined by modifying the definition of the
limit in the category theory.
The diagrams for the limit were restricted them in $QP \subset QW$, 
and the cones to define the limit were restricted to only 
those with vertices in the Poisson algebra.
Like the limit of ordinary number sequences, the naive classical limit and the 
weak classical limit may not have a limit, 
depending on how the sequence of objects is chosen.
On the other hand, the strong classical limit was defined for quantizations 
when a Poisson algebra is fixed. 
In contrast to the other classical limit, 
there is no flexibility in the choice 
of the classical limit, which is automatically determined 
when the Poisson algebra is chosen. 
Also, we found the fixed Poisson algebra is always the strong classical limit.
Finally, we discussed the inverse problem of determining 
the classical limit from some 
noncommutative Lie algebra.
From a fixed Lie algebra, we constructed a sequence of representations 
of the Lie algebra with relations, from which we constructed a Poisson algebra. Next, by constructing a sequence of quantization maps from the Poisson algebra,
 we obtain the sequence of representations of the Lie algebra 
as a sequence of $QW$ objects. 
A method to obtain this series of procedures 
from the principle of least action was proposed. 
The proposed method is just an example of how to approach to 
the inverse problem 
in the framework of $QW$, which is a
generalization of the fuzzy sphere case.
A more precise discussions on how generalizations 
should be made are needed in the future.
The challenge of obtaining classical manifolds 
from solutions of noncommutative manifolds of matrix models 
that actually describe string theory or M-theory is a future problem, too.
\\

As for $QW$, it is important to investigate its properties 
from a purely mathematical point of view as well, 
since it is a noncommutative geometric object that naturally follows 
from the overall formulation of the quantization of all Poisson algebras.
There are many interesting problems around $QW$ 
to understand the whole picture of noncommutative geometry.
These are issues to be clarified in the future.

%
%
\section*{Acknowledgements}
\noindent 
A.S.\ was supported by JSPS KAKENHI Grant Number 21K03258.
The author is grateful to A. Yoshioka for his detailed discussions 
and useful comments throughout the entire article.
The author also thanks the participants in the workshop
``Discrete Approaches to the Dynamics of Fields and Space-Time" 
for their useful comments.
We would like to thank T. Asakawa and G. Ishiki
for important information, helpful discussions and comments on the manuscript.

\newpage

\appendix
\section{Definition of $\tilde{O}(z^{1+\epsilon} )$}\label{ap1}
Since we have not defined a norm for objects of category $QW$ 
in this paper, Landau symbol $O$ does not make sense. 
So, we define an order $\tilde{O}$ by $x\in \mathbb{R}$ with the Euclidean norm.

\begin{definition}
Let $\mathcal{M}$ be an $R$-module for a commutative algebra $R$ over $\mathbb{C}$. Let $f_i$ be a complex valued continuous function such that
\begin{align*}
\lim_{x\to 0}  \frac{f_i(xz)}{x}  =0 ,
\end{align*}
where $x\in \mathbb{R}$ and $z\in \mathbb{C}$. For $a_i\in \mathcal{M}$ which is independent of $z\in \mathbb{C}$, we denote the element described as $\sum_i f_i(z)a_i\in \mathcal{M}$ by $\tilde{O}(z^{1+\epsilon} )$.
\end{definition}
From this definition, the term of $\tilde{O}(\hbar^{1+\epsilon})$ in $(\ref{lie})$ is also $0$ when $\hbar=0$. Note that   $\hbar$ itself is not necessarily continuous.

We use the following fact:
\begin{proposition}\label{propA2}
Let $t_i : A \rightarrow M_i$ be a weak quantization, and let $h_{ij}: M_i \rightarrow M_j$ be an $R$-algebra homomorphism.
Then 
$$
h_{ij} (\tilde{O}(\hbar^{1+\epsilon} (t_i )))= \tilde{O}(\hbar^{1+\epsilon} (t_i)) \in M_j.
$$
\end{proposition}
\begin{proof}
For  $\tilde{O}(\hbar^{1+\epsilon} (t_i ))=\sum_i f_i(z)a_i\in {M}_i $
\begin{align*}
h_{ij} (\tilde{O}(\hbar^{1+\epsilon} (t_i )))&=
h_{ij} ( \sum_i f_i(\hbar )  a_i ) 
=  \sum_i f_i(\hbar ) h_{ij}( a_i ) = \tilde{O}(\hbar^{1+\epsilon} (t_i)) \in M_j.
\end{align*}
\end{proof}

\newpage
\section{List of Symbols}\label{ap2}
In this paper, many symbols including new symbols are used.
So, the list of them is useful to read through. \\

\begin{table}[h!]
  \caption{List of Symbols}
  \label{table:data_type}
  \centering
  \begin{tabular}{|c|c|}
    \hline
$R$ & commutative ring over ${\mathbb C}$ \\
\hline
$\mbox{$R$-Mod}$ & category of $R$-module 
\\
\hline
$\mbox{$R$-alg}$ & category of $R$-algebra 
\\
\hline
${\mathscr Poisss}$ & 
\begin{tabular}{c}
category of Poisson algebras \\
whose morphisms are surjective Poisson morphism
\end{tabular}
\\ \hline
$wkQ$ & set of weak quantization maps ( including maps with $\hbar =0$)
\\ \hline
$Qh_x$ & set of weak quantization maps with $|\hbar |= x$
\\
\hline
$Q$ & set of weak quantization maps with $\hbar \neq 0$
\\
\hline
$Q( {{\mathscr Poisss}} )$ & 
$\{ t( q ) ~ | ~ q \in Q \}$
\\
\hline
$QP$ & 
\begin{tabular}{c}
subcategory of $R$-Mod  :
$ob (QP) := Q({\mathscr Poisss})$, \\
$Mor (t(q_i), t(q_j))$ is a set of $R$-algebra homomorphisms\\
 restricted their domain into $t(q_i)$
\end{tabular}
\\
\hline
${\mathscr Poisss} \bigsqcup QP$ & 
\begin{tabular}{c}
subcategory of $R$-Mod  :
$ob ( {\mathscr Poisss} \bigsqcup QP )=
ob ( {\mathscr Poisss} )  \bigsqcup  ob ( QP )$, \\
$Mor ( {\mathscr Poisss} \bigsqcup QP )=
Mor ( {\mathscr Poisss} ) \bigsqcup Mor ( QP )$
\end{tabular}
\\
\hline
$QW$ & 
\begin{tabular}{c}
subcategory of $R$-Mod  :
$ob (QW) :=ob ( {\mathscr Poisss} ) \cup  ob ( QP ) $, \\
$Mor ( QW ):= Mor ( {\mathscr Poisss} ) \bigsqcup Mor ( QP ) \bigsqcup Q$
\end{tabular}
\\
\hline
$U_P$ & forgetful functor from ${\mathscr Poisss}$ 
to $QW$ 
\\
\hline
$I_{QP}$ & identically immersion functor from $QP$ to $QW$
\\
\hline
$( U_P \downarrow I_{QP})$ & comma category of $U_P$ and $I_{QP}$
\\
\hline
$F $ & functor from $( U_P \downarrow I_{QP})$ to 
${\mathscr Poisss} \times R$-alg
\\
\hline
$P-QP$ & $P-QP := F( U_P \downarrow I_{QP})$ : ${\mathscr Poisss}-Q({\mathscr Poisss})$ 
quantization pair category
\\
\hline
$(A \downarrow I_{QP})$ &  category of all quantizations 
of $A$
\\
\hline
$F_A $ & functor from $ (A \downarrow I_{QP}) $ to $ R\mbox{-alg}$
\\
\hline
$D({\mathcal I})$ & Quantization Family with an index category ${\mathcal I}$
\\
\hline
$J^\bullet$ &  map from a category with total order to an index category
\\
\hline
$\mathcal{I}_{\hbar}(A)$ & $  J^\bullet (A \downarrow I_{QP})$
\\
\hline
$D_{\hbar}$ & diagram $\mathcal{I}_{\hbar}(A) \rightarrow (A \downarrow I_{QP})$\\
\hline
$A/QW_{\hbar}$ & 
\begin{tabular}{c}
category consists of $D_{\hbar} (\mathcal{I}_{\hbar}(A))$ and $\{( Id_A , A)\}$\\
with quantization maps 
whose source is $A$
\end{tabular}
\\
\hline
\end{tabular}
\end{table}


\section{Fuzzy Sphere and Fuzzy Torus }\label{appendix_fuzzy}
The  fuzzy sphere is considered in {\rm \cite{matrix1,fuzzy1}}. See {\rm \cite{matrix1,fuzzy1,fuzzyb,fuzzyc}} for details. 
In \cite{Rieffel:2021ykh}, more general and mathematically precise statements are given. 
Let $x^a~(1\le a\le 3)$ be coordinates of three-dimensional Euclidean space and $S^2$ be the two-sphere given by
\begin{align}\label{sphere}
\delta_{ab}x^ax^b=1.
\end{align}
Let $P[x^a]$ be the algebra of polynomials generated by $x^a$. For an ideal $I$ which is generated by the relation $(\ref{sphere})$, we consider a Poisson 
algebra $\mathcal{A}=P[x^a]\slash I$ with Poisson bracket
\begin{align}
\{x^a,x^b\}=\epsilon^{abc}x^c. \label{eq.poi}
\end{align}
For arbitrary $f\in \mathcal{A}$ is given as
\begin{align*}
f=f_0+f_ax^a+\frac{1}{2}f_{ab}x^ax^b+\cdots
\end{align*}
where $f_{a_1\cdots a_i}\in \mathbb{C}$ is completely symmetric and trace-free. The morphism $t_1$ from $\mathcal{A}$ to $T_1=\mathbb{C}$ is defined by
$
t_1(f):=f_0
$
and the morphism $t_2$ from $\mathcal{A}$ to a subspace $T_2:= t_2( \mathcal{A} )$ of a matrix algebra $M_2={\rm Mat}_{2}$ is defined by
\begin{align*}
t_2(f)&:=f_0{\bf 1}_2+f_at_2(x^a),\quad \quad t_2(x^a):=\frac{\hbar_2}{2}\sigma^a
\end{align*}
where $\sigma^a$ is Pauli matrix and ${\bf 1}_k$ is a $k\times k$ unit matrix. $t_2$ means that the morphism gives a map from a polynomial to a $2\times 2$ matrix. In the case of $k\ge 2$, morphisms $t_k:\mathcal{A}\to T_k := t_k (\mathcal{A})\subset M_k={\rm Mat}_{k}$ are defined by
\begin{align*}
t_k(f)&:=f_0{\bf 1}_k+f_{a_1}t_k(x^{a_1})+\cdots +
\frac{1}{(k-1)!} 
f_{a_1\cdots a_{k-1}}t_k(x^{a_1}\cdots x^{a_{k-1}})\\
t_{k}(x^a)&:=\hbar_k J^a_{k}= X^a_{k}
\end{align*}
where $J^a_k$ are generators for the $k$-dimensional irreducible representation of $\mathfrak{su}(2)$, and for $i\ge 2$ each $t_i(x^{a_1}\cdots x^{a_{i-1}})$, which is fixed to be well-defined, is generated by $J^a_i$. The morphism $t_k$ gives a map from a polynomial to a $k\times k$ matrix. From the Casimir relation
\begin{align*}
J^a_kJ_{ka}=\frac{1}{4}(k^2-1){\bf 1}_k
\end{align*}
and $(\ref{sphere})$, the relationship between $k$ and $\hbar$ is given as
\begin{align}\label{relak}
\frac{4}{k^2-1}=\hbar^2_k.
\end{align}
From the commutation relation of $J^a_k$
\begin{align}
[J^a_k,J^b_k]&=i\epsilon^{abc}J^c_k, \quad \quad [t_k(x^a), t_k(x^b)]=i\hbar^2_k \epsilon^{abc}J^c_k=i\hbar_k \epsilon^{abc}t_k(x^c). \label{eq.com}
\end{align}
Using (\ref{eq.com}) we can show that $t_k \in Q $ $(k \ge 2 )$. 
Note that $\langle T_k \rangle = Mat_k = :M_k$.\\
\bigskip

Next, we consider a fuzzy torus in a similar way.
(See for example \cite{Bal:2004ai}. )
Let $(\theta_1 , \theta_2)$ be a coordinate of $T^2 = S^1 \times S^1$,
where $\theta_i \in [0, 2\pi)$.
A algebra 
$ \displaystyle \mathcal{B}= \{ f(\theta_1 , \theta_2) = \sum_{l_1, l_2} f_{l_1,l_2} y_{l_1,l_2} \}$,
where 
$ y_{l_1, l_2}:= \exp (\sum_i^2 \sqrt{-1}l_i \theta_i )$
 $( l_i \in {\mathbb Z}_{\ge 0} )$ are functions on $T^2$.  
Poisson bracket is defined by
\begin{align*}
\{ y_{l_1, l_2} , y_{m_1, m_2} \}= 
-\pi (l_1 m_2 -l_2 m_1) y_{l_1 + m_1, l_2+m_2} 
\end{align*}
The algebra of Fuzzy torus in $ M_k={\rm Mat}_{k}$ is 
generated by the clock matrix $U$ and shift matrix $V$
of size $k$ determined by a complex number $q= e^{2\pi i /k}$,
where
\begin{align*}
U= \left(
\begin{array}{cccc}
q^0 & {} &{} &{}\\
{} & q^1 &{} &{} \\
{} & {} & \ddots &{}\\
 {} &{} &{} & q^{k-1}
\end{array}
\right) , \qquad 
V= \left(
\begin{array}{cccc}
0 & 1 &{} &{}\\
{} & 0 &\ddots &{} \\
{} & {} & \ddots &1 \\
 1 &{} &{} & 0
\end{array}
\right) .
\end{align*}
$U, V$ satisfy $U^k = V^k = Id_k$ and $VU= q UV$.
By using these $U,V$, we introduce
\begin{align*}
Y_{m, n} := U^m V^n .
\end{align*}
Then they satisfy 
\begin{align*}
[Y_{l_1,l_2} , Y_{m_1 , m_2} ] = 
(q^{l_2 m_1} - q^{l_1 m_2})Y_{l_1 + m_1, l_2+m_2} .
\end{align*}
Note that 
$q^{l_2 m_1} - q^{l_1 m_2}= -\sqrt{-1}\frac{2\pi}{k}(l_1 m_2 -l_2 m_1)+O(1/k^2)$.
The quantization $q_k : \mathcal{B} \rightarrow Y_k := q_k(\mathcal{B})\subset  M_k$ is defined as the same manner with the fuzzy sphere:
\begin{align*}
q_k ( \sum_{l_1, l_2} f_{l_1,l_2} y_{l_1,l_2} ) = 
\sum_{l_1, l_2} f_{l_1,l_2} Y_{l_1,l_2} .
\end{align*}
Note that $\langle Y_k \rangle = M_k$.


\end{document}